\documentclass[12pt]{amsart}
\usepackage{amssymb}
\usepackage{amsmath}
\usepackage{amsfonts}
\usepackage{mathrsfs}
\usepackage{graphicx}
\usepackage{color}
\usepackage[onehalfspacing]{setspace}
\usepackage{caption}
\usepackage[round]{natbib}

\usepackage{enumerate}
\usepackage[utf8]{inputenc}
\usepackage{charter}

\usepackage[colorlinks=true,citecolor=blue,urlcolor=blue,pdfpagemode=UseNone,pdfstartview=FitH]{hyperref}
\usepackage{apptools}
\usepackage{chngcntr}
\usepackage{float}
\usepackage{multibib}
\newcites{main,supp}{References,References}

\makeatletter
\def\section{\@startsection{section}{1}
	\z@{1.0\linespacing\@plus\linespacing}{.8\linespacing}{\Large}}

\def\subsection{\@startsection{subsection}{2}
	\z@{.8\linespacing\@plus.7\linespacing}{.7\linespacing}{\large}}

\def\subsubsection{\@startsection{subsubsection}{3}
	\z@{.5\linespacing\@plus.7\linespacing}{-.5em}{\normalfont\bfseries}}
\makeatother

\numberwithin{equation}{section}

\newtheorem{theorem}{Theorem}[section]
\newtheorem{lemma}{Lemma}[section]
\newtheorem{proposition}{Proposition}[section]
\newtheorem{corollary}{Corollary}[section]
\theoremstyle{definition}
\newtheorem{definition}{Definition}[section]

\theoremstyle{definition}
\newtheorem{assumption}{Assumption}[section]

\theoremstyle{definition}

\title{}
\begin{document}

%% START TITLE PAGE

	\vspace*{3ex minus 1ex}
	\begin{center}
		\Large \textsc{Schooling Choice, Labour Market Matching, and Wages}

		\bigskip
	\end{center}
	
	\date{%
		%TCIMACRO{\TeXButton{Today}{\today}}%
		%BeginExpansion
		\today%
		%EndExpansion
	}
	
	\vspace*{3ex minus 1ex}
	\begin{center}
		Jacob Schwartz\\
		\textit{Department of Economics, University of Haifa}\\
		\medskip
		
	\end{center}
	
	\thanks{Contact: Jacob Schwartz, Department of Economics, University of Haifa, 199 Aba Khoushy Ave., Mount Carmel, Haifa, 3498838, Israel. Email: jschwartz@econ.haifa.ac.il. }
	
	%		\fontsize{13}{14} \selectfont
	
	 	\begin{abstract}
       We develop inference for a two-sided matching model where the characteristics of agents on one side of the market are endogenous due to pre-matching investments. The model can be used to measure the impact of frictions in labour markets using a single cross-section of matched employer-employee data.  The observed matching of workers to firms is the outcome of a discrete,  two-sided matching process where firms with heterogeneous preferences over education sequentially choose workers according to an index correlated with worker preferences over firms.  The distribution of education arises in equilibrium from a Bayesian game: workers, knowing the distribution of worker and firm types, invest in education prior to the matching process. Although the observed matching exhibits strong cross-sectional dependence due to the matching process, we propose an asymptotically valid inference procedure that combines discrete choice methods with simulation.
      
		\medskip
		
		\medskip
		
	\end{abstract}
	
	\maketitle
	\pagebreak
	%% END TITLE PAGE

\section{Introduction}
Since the 1980s, economists have attributed rising wage inequality
to a number of sources. One possible source of such inequality is
positive assortative matching between workers and firms - the tendency
for the quality of workers and firms who match with one another to
be positively correlated.\footnote{Recent empirical papers examining the role of sorting on wage inequality
include \citet*{Card/Heining/Kline:13:QJE}, \citet*{Barth/Bryson/Davis/Freeman:2016:JL},
and \citet*{Kantenga/Law:2016}.} Unfortunately, studying matching in labour markets presents a serious
challenge when the decisions of individual job seekers affect each
other's hiring outcomes. This paper develops a methodology to address
this challenge. In particular, we show how cross-sections of matched
employer-employee data can be used to study the role that a labour
market matching technology plays in shaping the equilibrium distributions
of education and wages. The structural model we develop can capture assortative matching between workers
and firms even in the absence
of complementarities between worker and firm types in the match production
function.\footnote{The value of a match between any type of worker, $h$, and any type
of firm, $k$, can be represented using a positive, increasing function,
$f(h,k)$. We say that the types are complements in $f$ when the
marginal product of an $h$ type is higher when matched with a higher
$k$ type (and vice versa). }

A general overview of the labour market in the model is as follows.
Agents from one side of the market sequentially choose agents from
the other side according to their preferences. Preference rankings
of the choosers depend on a preference parameter, along with the capital
of both types of agents. The order in which the choosers pick depends
on the chooser's capital and a matching technology parameter. Before
matching, the agents who will be chosen are allowed to simultaneously
decide their capital given the distribution of the chooser's capital
and the underlying parameters of the economy (including the frictions).

This paper contributes to the econometric literature concerned with
inference in two-sided matching models.\footnote{See \citet*{Chiappori/Salanie:16:JEL} for a review of this literature.
A seminal paper in this literature is \citet*{Choo/Siow:06:JPE},
which considers inference in a transferable utility setup with a continuum
of agents. } A key feature of our setup is that the characteristics of agents on
one side of the market are endogenous - in particular, arising in
equilibrium from a pre-matching investment game. We show how, rather
than making the empirical analysis intractable, accommodating such
pre-matching investments provides the researcher useful information
for inference.\footnote{A popular approach for estimating two-sided matching models builds
on the notion that the observed matching is pairwise stable. For example,
see \citet*{Fox/Bajari:13:AEJM}, \citet*{Echenique/Lee/Shum:13:SEM},
and \citet*{Menzel:15:Ecta}. Requiring that the observed matching
be pairwise stable may be unrealistic in the context of a frictional
labour matching market of the sort that is the focus of this paper.} We propose a two-stage approach for inference on the agents' preferences
and the matching technology. In the first stage, we fix the matching
technology and construct confidence regions for the preference parameter
by estimating the Bayesian game associated with the workers' pre-match
investment in education decision. We show that this problem can be
cast in a discrete choice framework yielding tractable and consistent estimation via maximum likelihood when the workers' educational decision takes one
of two values (college, or no college). In the second stage, we construct
confidence intervals for the matching technology using a simulation-based
inference approach. In the first stage, the presence of the matching
function in workers' expected utility function makes estimating workers'
equilibrium expectations highly non-trivial. Nevertheless, under reasonable
assumptions, we show that workers' equilibrium expectations can be
written in a closed form suitable for consistent estimation. The second-stage
inference on the matching technology uses the following insight: once
the matching process is specified, the finite sample distribution
of the observed matching is known up to a parameter.\footnote{This idea of using a structural model to characterize the joint distribution
of a discrete matching model that can then be used for inference on
the model parameters builds from \citet*{Kim/Schwartz/Song/Whang:19:Econometrics}.
Although computationally intractable when the dimension of the parameter
is large, this approach is attractive for inference on the matching
technology parameter in the second stage of our approach.} We construct a test statistic that measures the distance between the
observed joint distribution of worker education and matched firm capital
to simulated counterparts. A confidence interval for the matching
technology can then be constructed by inverting the test.

This paper builds on the fundamental insights of \citet*{Becker:73:JPE}
and \citet*{Gale/Shapley:62:AMM} to illustrate how an economically meaningful notion of sorting can be captured in a model that assumes additive worker and firm effects. Since \citet*{Abowd/Kramarz/Margolis:99:Ecta}
(AKM), the availability of matched employer-employee data has allowed
researchers to study the role that unobserved worker and firm attributes
play in driving wage variation over time. In AKM, the correlation
of worker and firm fixed effects from wage regressions is taken to
capture a notion of sorting. Although popular for investigating the
wage structure, a burgeoning literature has criticised the viability
of AKM for detecting sorting on unobservables. In particular, the
additive structure of AKM implies that wages are monotone in firm
type - an implication that is difficult to reconcile with equilibrium
models of sorting with and without frictions (\citet*{Eeckout/Kircher:11:ReStud},
\citet*{LopesDeMelo:18:JPE}).\footnote{\citet*{Gautier/Teulings:06:JEEA} was an early empirical study that
detected a concave relationship between wages and firm type.} For example, in \citet*{Eeckout/Kircher:11:ReStud}, a low-type worker
can receive a lower wage at a high-type firm since the worker must
implicitly compensate the high-type firm in equilibrium for forgoing
the opportunity to fill a vacant job with a higher-type worker.\footnote{There are other reasons wages may be non-monotonic in firm type. In
\citet*{Postel-Vinay/Robin:2002:Ecta}, workers may be willing willing
to accept lower wages at higher type firms when they expect to receive
higher wages in the future. } 

 In his
seminal 1973 paper on the marriage market, Becker argued that when
the match production function is supermodular\footnote{When $f$ is differentiable, (strict) supermodularity is equivalent
to $\partial^{2}f(h,k)/\partial h\partial k>0$. } and utility is transferable
between matched agents, high types can outbid low types for the best
partners, leading to an equilibrium with positive assortative matching. In the same paper, Becker noted that sorting can also arise in an non-transferable
utility (NTU) framework when the payoffs of the agents on both sides
of the market are monotonic in the other agent's type. To explain why, Becker invokes the notion of pairwise stability \citet*{Gale/Shapley:62:AMM}.\footnote{This insight comes to us by way of \citet*{Chade/Eeckhout/Smith:17:JEL}'s
excellent review of the search and matching literature. } To illustrate, consider an economy with four agents where a low-type
firm is paired with a high-type worker and vice-versa. Such a matching
is unstable when high types are preferred, because both high-types
will agree to abandon their low-type partners for one another. NTU
arises naturally in our model from the assumption that wages for any
matched pair are determined exogenously (in fact, by a Nash bargaining
assumption). In a special case of our model in which the preferences of agents are indeed
monotonic, sorting - and some inequality - may emerge. In this case,
complementarities are not necessary for sorting but merely amplify
the effects of sorting, since interactions between worker and firm
types in the wage function lead to more wage dispersion than when
such interactions are absent. 

Search and matching models have emerged as the leading alternative
to the AKM framework for studying sorting in labour markets.\footnote{\citet*{Hagedorn/Law/Manovskii:17:Ecta}, \citet*{Bagger/Lentz:2018:ReStud},
\citet*{Lise/Meghir/Robin:16:RED} and \citet*{LopesDeMelo:18:JPE}
all find evidence of positive sorting when an AKM approach finds negligible
sorting.} In this literature, the standard matching technology is one that
converts aggregates of vacancies and unemployed workers into matches.
Although treating matching at the aggregate level simplifies the analysis
considerably, any strategic interdependence that may be present in
the matching process is assumed away (\citet*{Chade/Eeckhout/Smith:17:JEL}).
One contribution of this paper is to develop and estimate a model
that takes such strategic interdependence in the matching process
seriously. Capturing such interdependence is desirable, since in many labour markets the impact of an individual's decisions on the outcomes of other workers are highly relevant and cannot be ignored.\footnote{For example, a worker's decision to get a master's degree in finance
will not only affect the likelihood that he gets a job at an investment
bank, but also the likelihood that his competitors get the job.} In the equilibrium of our model, (and in contrast to standard search models), the probability that a
worker matches to a given firm typically depends on the decisions
of all the other agents in the economy. Another key facet of search
models is that workers direct their job search based on the wages
that employers set for them. However, recent studies of online
job markets have found that it is relatively uncommon for positions
to explicitly post wages.\footnote{For example, \citet*{Marinescu/Wolthoff:16:WP} study the role of
job titles in directing the search of workers report that only 20\%
of job the advertisements CareerBuilder.com report a wage. } Another way in which this paper differs from the traditional search
literature is that we do not require workers do not observe posted
wages directly. Instead, workers know the underlying distributions
of job characteristics and the matching process prior to simultaneously
investing in education. In this sense, the worker's decision to invest
in education is the channel by which workers are able to direct their
search.

The framework in this paper supposes that each equilibrium gives rise to a single
large matching between workers and firms.\footnote{This contrasts with cases in which the researcher sees many independent
copies of games involving few players, such as those studied by \citet*{Bresnahan/Reiss:91:JPE},
\citet*{Ciliberto/Tamer:09:Ecta}, \citet*{Berry:92:Ecta} and many
others. See \citet*{Xu:18:IER}, \citet*{Song:14:WP}, \citet*{Menzel:16:ReStud}
for more papers discussing the estimation of large Bayesian games. } Under familiar assumptions (e.g., iid and separable private information),
we follow similar arguments to \citet*{Aguirregabiria/Mira:19:WP}
to prove that an equilibrium exists. The setup here, however, also allows
us to provide sufficient conditions for equilibrium uniqueness. 

This paper is also part of the literature concerned with estimating
cross-sectionally dependent observations.  In our setup, the observed
matching of workers to firms exhibits cross-sectional dependence of
an unknown form due to the matching process. This means that asymptotic
inference approaches that appeal to the the law of large numbers and
central limit theorems will not work. The approach we pursue builds
on \citet*{Kim/Schwartz/Song/Whang:19:Econometrics} which discusses
how inference in structural matching models are possible when knowledge
of the matching process can be used to characterize the joint distribution
of the observed matching. This paper shows how such a simulation-based
inference approach, cumbersome when the dimension of the parameter
space is high dimensional or complex, is useful for estimating a subset
of the parameters in structural models with cross-sectional dependence.\footnote{The simulation-based approach used in the second-stage of the inference
procedure is known as a Monte Carlo test. Monte Carlo tests have a
history in econometrics dating back at least to the 1950s, as discussed
by \citet*{Dufour/Khalaf:01:TE} in their overview of the technique.}

Section \ref{subsec:Some-Implications-Of} provides intuition on how our model captures the relationship between matching frictions, sorting, and inequality. In particular, we illustrate how a fall in matching frictions can yield two opposing impacts on wage inequality via their effects on sorting and the supply of highly educated workers. In the model parameterization considered, a fall in frictions leads to both an increase in the equilibrium supply of highly educated workers and an increase in positive assortative matching between workers and firms. That is, the latter sorting effect increases wage inequality while the former supply effect acts in the opposite direction.  In general, the impact of matching frictions on wage inequality are more pronounced when worker and firm types are complements in the match production function.
The section also illustrates how a fall in information frictions can lead to a dramatic rise in
the education wage premium through sorting while at the same time, a much more modest
increase in the supply of highly educated workers. Thus, changes in
informational frictions may be a useful way to explain a puzzling
empirical findings concerning the relationship between wage premia
and educational attainment.\footnote{See \citet*{Card/Lemieux:01:QJE}.} 

Section \ref{sec:section_two_structural_model} introduces the model
of two-sided labour market matching with frictions. In the baseline
model of Section \ref{sec:baseline}, workers and firms with exogenous
characteristics match with one another and split the match surplus
according to a Nash bargaining rule. The rest of the paper is organized
as follows. Section \ref{sec:StructuralModel}
extends the baseline model to allow for endogenous worker characteristics
- after observing their type, workers simultaneously invest in education
prior to entering the labour market. Section \ref{sec:section_3_econometric_inference}
outlines an approach for inference on the parameters of the structural
model of Section \ref{sec:StructuralModel}. Section \ref{sec:MonteCarlo} presents a small simulation study illustrating the finite sample size and power  performance of the first-stage inference. Section \ref{sec:conclusion_c1} concludes. Mathematical proofs are confined to two appendices: Appendix A contains results concerning the existence and uniqueness of the equilibrium of the game described in Section \ref{sec:StructuralModel} while Appendix B contains supplemental results relating to the first-stage inference on preferences.

\begin{singlespace}

\section{The Labour Market As a Two-Sided Matching Market\label{sec:section_two_structural_model}}
\end{singlespace}

Our goal is to study the distribution of education and wages using
separate cross-sections of matched employer-employee data. The first
subsection introduces the core elements of the model that will serve
as the basis for the structural model in the second subsection. 

\subsection{Baseline model}\label{sec:baseline}

Let $N_{h}=\{1,...,n_{h}\}$ be the set of workers and $N_{f}=\{1,...,n_{f}\}$
be the set of firms, where $n_{h}$ and $n_{f}$ are used to denote
the total number of workers and firms, respectively. Each worker seeks
one job and each firm seeks to hire one worker.

The matching of workers to firms will be determined by the preference
rankings of workers and firms. Workers value the capital of firms,
$K=(K_{j})_{j\in N_{f}}$, and firms value the human capital of workers
$H=(H_{i})_{i\in N_{h}}$, where $K_{j}$ and $H_{i}$ are scalars.
Any worker $i$ who is matched with firm $j$ receives wage $w_{ij}\geq0$
while firm $j$ receives profit $\rho_{ji}\geq0$, where both wages and profits may also depend on a parameter,  $\theta\in\mathbf{R}^{d}$.\footnote{In this setup, $\theta$ represents the preferences of both workers
and firms. As we will see, $w_{ij}$ and $\rho_{ji}$ depend on the
output of worker $i$ at firm $j$, and the production function that
gives rise to this output will depend on a part of $\theta$. } Since our framework supposes that wages and profits are always non-negative
for any worker and firm that could match, we will assume throughout
the paper that no agent will ever unilaterally dissolve a match to
become unmatched. This requirement that any matching satisfy an individual
rationality constraint is embodied in the following condition:\footnote{The current setup is tailored to settings where the researcher has
at least one cross-section of matched employer-employee data and the
agents who are unmatched are not of primary interest in the analysis.
An interesting (and challenging) extension of the current framework
would accommodate the possibility of unmatched agents, and hence unemployment.}

\bigskip{}

\noindent \textbf{Condition} \textbf{IR} (Individual rationality of
matches): \textit{For each $i\in N_{h}$, and $j\in N_{f}$ $w_{ij}\geq0$
and $\rho_{ji}\geq0$. }

\bigskip{}

Based on the values of $\rho_{j}=(\rho_{ji})_{i\in N_{h}}$ each firm
$j$ can construct preference rankings over the workers. We suppose
that if the firm is ever indifferent between one or more workers,
then the firm picks preference rankings over these workers at random.
Next, we introduce a condition on the worker's wage function that
will grant us a natural economic interpretation of the matching process (along with our
notion of information frictions). \bigskip{}

\noindent \textbf{Condition} \textbf{H} (Homogeneous worker preferences):
\textit{For each $i\in N_{h}$, the wage of worker $i$ is  increasing in the capital of their matched firm.}%\geq %w(H_{i},K_{2})$
%for all $K_{1}\geq K_{2}$. }

%\noindent \textbf{Condition} \textbf{H} (Homogeneous worker preferences):
%\textit{For each $i\in N_{h}$, $w(H_{i},K_{1})\geq %w(H_{i},K_{2})$
%for all $K_{1}\geq K_{2}$. }

\bigskip{}

The condition is tantamount to a notion of worker preference homogeneity,
implying that all workers prefer higher capital firms. Supposing that workers accurately observe the capital of firms, the condition  implies that a matching algorithm in which the highest capital firm,
$j_{1}$, choose his preferred worker, $i_{1}\in N_{h}$, the second
highest capital firm, $j_{2}$, choose his preferred worker $i_{2}\in N_{h}\backslash\{i_{1}\}$
and so on is an example of the \textit{serial dictatorship} mechanism and would produce a stable matching.\footnote{\citet*{Satterthwaite/Sonnenschein:81:ReStud}. See Section 2.2. of \citet*{Roth/Sotomayor:92}. } 

In order to build a model that accounts for the possibility of mismatches
between workers and firms, we suppose that information frictions are present in the market. Specifically, we suppose that workers do not directly observe realizations of the firm's capital. Instead, each worker sees $v=(v_{j})_{j\in N_{f}}$, where $v_{j}$
is a `noisy' measure of firm $j$'s capital. In particular, suppose that
workers see 
\begin{equation}
v_{j}=\beta K_{j}+\eta_{j},\label{eq:section_two_v_functional_form}
\end{equation}
for each $j$, where $\beta\in B$, $B\subset\mathbf{R}$ is the
parameter space of $\beta$, and $\eta_{j}$
is a random variable that is independent across $j$. The size of
the variance of $\eta_{j}$ relative to the magnitude $\beta$ represents
the magnitude of information frictions in the matching process. It is clear
that when $\beta$ is zero and the variance of $\eta_{j}$ is positive,
then this setup yields random matching from firm to worker the characteristics,
since variation across firm capital plays no role in determining the
realizations of $v$. Furthermore, when $\beta\neq0$ and $\text{Var}(\eta_{j})=0$,
it will be as if firm capital is observed by the worker, since $v_{j}$
is determined entirely by the firm's capital. In the latter case,
when $\beta>0$ workers would favour firms with the largest realizations
of $v$, while in the case that $\beta<0$, workers would favour firms
with the smallest realizations of $v$. However, even in the case
that $\text{Var}(\eta_{j})>0$, $v_{j}$ still conveys some useful
information to the worker under certain circumstances. To see this, suppose that $v_{j}$'s follow equation \ref{eq:section_two_v_functional_form} with
$\beta>0$ and let $\eta_{j}$'s be iid\footnote{We will impose such an assumption in a later section.}. Then, any worker who sees $v_{j1}$
exceed $v_{j2}$ will prefer matching with Firm $j_{1}$ over Firm
$j_{2}$, since the worker recognizes that the distribution of $K_{j1}$ conditional on $v_{j1}=\tilde{v}_{j1}$
stochastically dominates the distribution of of $K_{j2}$ conditional
on $v_{j2}=\tilde{v}_{j2}$ when the worker observes $\tilde{v}_{j1}>\tilde{v}_{j2}$.

The following condition specifies the matching process we will use
throughout the paper. 

\bigskip{}

\noindent \textbf{Condition SD }(Matching process): \textit{The matching
of workers to firms in the economy arises as follows. The highest
$v$ firm, $j_{1}$, chooses his preferred worker, $i_{1}\in N_{h}$,
the second highest $v_{2}$ firm, $j_{2}$, chooses his preferred
worker $i_{2}\in N_{h}\backslash\{i_{1}\}$, and so on, until the
lowest $v$ firm, $j_{n_{f}}$, chooses his preferred worker among
those not chosen by any higher ranked firms.}

\bigskip{}

One way of understanding this matching algorithm in economic terms
is to consider the following thought experiment. Imagine a situation
in which a group of job-seekers have assembled in a large room on
the day of a job fair. Workers do not observe the true quality of
any of the firms, (represented by $K$), but they do see each firm's
value of $v$. When $\beta>0$ and $\eta_{j}$'s are iid, each worker
is happiest to match with the highest $v$ firm, since the distribution
of capital associated with the highest $v$ firm stochastically dominates
the distribution of capital associated with any of the lower $v$
firms. A procedure in which the highest $v$ firm, $j_{1}$, chooses
his preferred worker, $i_{1}\in N_{h}$, the second highest capital
firm, $j_{2}$, chooses his preferred worker $i_{2}\in N_{h}\backslash\{i_{1}\}$
and so on, will have no complaints from any of the participants at
the job fair \textendash{} that is, until uncertainty associated with
$K$ is revealed. In this world, agents will typically have more regret
(and hence a greater desire to rematch) when the frictions in $v$
are large. However, rematching is outside the scope of the model. 

Next, we add some further structure to wages and profits. In particular,
we will assume that the payoffs for any two matched agents follow
a Nash bargaining structure.  Let $\tau\in(0,1)$ be the bargaining weight. A worker $i$ who matches
with a firm $j$ receives
\begin{eqnarray}
w_{ij} & = & \tau f(H_{i},K_{j})+(1-\tau)g(H_{i})\text{ and}\label{eq:preferences}\\
\rho_{ji} & = & (1-\tau)\left(f(H_{j},K_{j})-g(H_{i})\right),\nonumber 
\end{eqnarray}
where $f$ is the worker-firm output function and $g(H_{i})$ is an
outside option function, both of which may depend on elements of $\theta$. In a subsequent section, we will allow worker covariates,
$X_{i}$, to effect wages through the outside option function, $g.$\footnote{$X_{i}$'s have support $\mathcal{X}\subset\mathbf{R}^{d}$, where
$d$ is an integer greater than or equal to one.} The following condition requires $f$ to satisfy some intuitive properties
with respect to the worker and firm capital variables. 

\bigskip{}

\noindent \textbf{Condition F} (Production function):\textit{ $f$
is increasing in human capital and firm capital.}

\bigskip{}

Condition F merely requires that more capital leads to more output
- it does not impose that the worker and firm attributes be complements
in $f$. Section \ref{sec:StructuralModel}
goes into further detail about the role of $f$ in this model. 

\bigskip{}

\subsection{Frictional Matching Model with Worker Investments\label{sec:StructuralModel}}

We now introduce a structural model where workers simultaneously invest
in education prior to the serial dictatorship matching process as
outlined in the previous section. A general overview of the matching
process is as follows: i) workers, observing only their type, simultaneously
choose a level of education, ii) $v$ is realized, iii) firms, seeing
only the education of workers, match according to Condition SD. 

Although firms select their preferred workers in the serial dictatorship
phase after constructing preference rankings over the workers, firms
are not considered strategic agents within the context of the investment
game itself.  

There are $n_{h}$ players indexed by $i\in N_{h}$. Each player chooses
an education level, $h_{i}$, from the discrete set $\mathcal{H}\equiv\{1,...,J\}$
to maximize their expected payoff. Let $\lambda=(\theta',\beta)$,
where $\beta$ is the matching frictions parameter and $\theta\in\mathbf{R}^{d}$
is a preference parameter. The payoff function of player $i$ comprises
the wage less a cost of education,
\begin{equation}
u(h_{i},h_{-i},x_{i},k,\eta,\varepsilon_{i};\lambda)=\omega(h_{i},h_{-i},x_{i},k,\eta;\lambda)-c(h_{i},x_{i},\varepsilon_{i};\lambda),\label{eq:payoff_function}
\end{equation}
where $h_{-i}\in\mathcal{H}_{-i}$ are the choices of the other agents,\footnote{Since the set of pure strategies for each agent is $\mathcal{H}$,
it follows that $\mathcal{H}_{-i}=\mathcal{H}^{n_{h}-1}$ for each
$i$, where $\mathcal{H}^{n_{h}-1}$ denotes the $(n_{h}-1)$-ary
Cartesian power of $\mathcal{H}$. } $x_{i}\in\mathcal{X}$ and $\varepsilon_{i}\in\mathbf{R}^{J}$
are the private information of worker $i$, and $k\in\mathbf{R}^{n_{f}}$
and $\eta\in\mathbf{R}^{n_{f}}$ are vectors of exogenous firm variables
that are unobserved by the workers. Although $\varepsilon_{i}$ and
$x_{i}$ are private information of the worker, we will assume $x_{i}$
is observed by the econometrician in a subsequent section. The variable
$\varepsilon_{i}$ represents the worker's private cost associated
with each of the $J$ education levels. In Section \ref{subsec:3_2_First-stage-estimation-of_theta}
we will supply explicit assumptions on worker and firm information
that illustrates why, given the matching process, the components of
the payoff function depend on model's underlying variables in the
way stipulated by equation (\ref{eq:payoff_function}). 

We now provide additional conditions that establish the existence
of a Bayesian Nash equilibrium for our game (which we prove in Section
\ref{sec:equilibrium}). 
\noindent \begin{assumption}\label{AssIndep}
\noindent \textit{(a) $K_{j}$'s, $\eta_{j}$'s are independent across $j$. $X_{i}$'s, $\varepsilon_{i}$'s are independent across $i$. $X$, $K$, $\varepsilon$, and $\eta$ are independent. (b) $\varepsilon_{i}$'s are continuously distributed.}
\end{assumption}\begin{assumption}\label{AssSep}\textit{The
cost function is separable in private information}: 
\[
c(h_{i},x_{i},\varepsilon_{i};\lambda)=c_{0}(h_{i},x_{i};\lambda)+\varepsilon_{i}'d(h_{i}),
\]
\textit{ where $d(h_{i})$ is a $J$-dimensional vector with one in
the $h_{i}$-th row and zero otherwise.} 
\end{assumption}

The assumptions of separability and independence are common in the structural
literature.\footnote{For example, see the discussions in \citet*{Kasahara/Shimotsu:08:JoE} and \citet*{Xu:14:EJ}.} In Section \ref{sec:equilibrium},
we show that Assumptions \ref{AssIndep} and \ref{AssSep} are sufficient for establishing the
existence of the Bayesian Nash equilibrium for the game of this section.
For now, we will provide some intuition into the worker's education
decision problem. First, we define the set of pure strategies as
$\sigma=\{\sigma_{i}(x_{i},\varepsilon_{i}):i\in N_{h}\}$ where $\sigma_{i}$
is a function that maps from $\mathcal{X}\times\mathbf{R}^{J-1}$
into $\mathcal{H}$. Assumption \ref{AssSep} says that we can write the expected
utility of agent $i$ with covariates $x_{i}$, who chooses $h_{i}$
under beliefs $\sigma$ as 
\begin{eqnarray}
U_{i}(h_{i},x_{i},\sigma,\varepsilon_{i}) & = & \tilde{U}_{i}(h_{i},x_{i},\sigma)+\varepsilon_{i}'d(h_{i}),\label{eq:EU}
\end{eqnarray}
where the first term in the expected utility is 
\begin{eqnarray}
\tilde{U}_{i}(h_{i},x_{i},\sigma) & = & \sum_{h_{-i}\in\mathcal{H}_{-i}}\tilde{u}_{i}(h_{i},h_{-i},x_{i})P_{-i}(h_{-i}|\sigma),\label{eq:FTEU}
\end{eqnarray}
and
\begin{equation}
\tilde{u}_{i}(h_{i},h_{-i},x_{i})\equiv\tilde{\omega}_{i}(h_{i},h_{-i},x_{i})-c_{0}(h_{i},x_{i}),\label{eq:expected_utility_fn}
\end{equation}
where $\tilde{\omega}_{i}(h_{i},h_{-i},x_{i})$ is given by
\begin{align*}
\tilde{\omega}_{i}(h_{i},h_{-i},x_{i}) & =\mathbf{E}[\omega(H_{i},H_{-i},X_{i},K,\eta;\lambda)|H_{i}=h_{i},H_{-i}=h_{-i},X_{i}=x_{i}],
\end{align*}
 and expectation is taken with respect to the distributions of $K$
and $\eta$.  By Lemma (\ref{lemma_1_expected_utility_representation})m we can rewrite equation \ref{eq:FTEU} as
\begin{eqnarray*}
\tilde{U}_{i}(h_{i},x_{i},\sigma) & = & \sum_{h_{-i}\in\mathcal{H}_{-i}}\tilde{u}_{i}(h_{i},h_{-i},x_{i})\prod_{j\in N_{h}\backslash\{i\}}P_{j}(h_{j}|\sigma_{j}).
\end{eqnarray*}

Throughout this paper, we will consider the case in which the wages
of workers are determined by Nash bargaining. As in equation \ref{eq:preferences},
we will suppose that firm capital only enters the worker's payoff
through the production function.  Denote $\mathcal{M}(i)$ as the identity of the firm that worker $i$ matches to as a result of the matching process, and $K_{\mathcal{M}(i)}$ as the level of capital associated with firm $\mathcal{M}(i)$. Under these assumptions, we may write $\tilde{\omega}_{i}(h_{i},h_{-i},x_{i})$
as 
\begin{eqnarray} \tilde{\omega}_{i}(h_{i},h_{-i},x_{i}) & = & \tau\tilde{f}_{i}(h_{i},h_{-i})+(1-\tau)g(h_{i},x_{i}),\label{eq:expected_wage} \end{eqnarray}where
\begin{align*}
\tilde{f}_{i}(h_{i},h_{-i}) & =\mathbf{E}[f(H_{i},K_{\mathcal{M}(i)})|H_{i}=h_{i},H_{-i}=h_{-i},X_{i}=x_{i}],
\end{align*}
the expectation is taken with respect to the distributions of
$K$ and $\eta$, and we have allowed the worker's characteristics to enter the payoff
function through the outside option function, $g$.\footnote{Here, $\tau_{i}=\tau$ for each $i$. The framework here can be extended to incorporate heterogeneity in worker
bargaining positions.   \citet*{Bagger/Lentz:2018:ReStud} emphasize the importance of endogenous search intensity and matching variation
(e.g., \citet*{Postel-Vinay/Robin:2002:Ecta})  in understanding the causes of wage inequality.}

Education affects the worker's expected utility in a number of ways.
The first two are obvious: since $f$ is increasing in $h_{i}$ by
Condition F, the worker who invests in a higher level of education
obtains a higher wage at any firm he matches to. The worker's choice
of education also affects his payoff through the outside option function,
$g$. The novel channel in this setup is that $h_{i}$ also determines
the expected quality of the firm that $i$ matches to. Even though
(as mentioned before) firms in this model are non-strategic agents,
the functional form of the production function, $f$, plays a key
role in determining whether or not firms with different levels of
capital exhibit different preferences for workers of differing levels
of education. To see how $f$ determines whether or not firms' preferences
are heterogeneous, consider the Nash bargaining preferences of a firm
for any worker who chooses education level $h$: 
\begin{equation}
\rho(k,h;\theta)=(1-\tau)(f(h,k;\theta)-\tilde{g}(h;\theta)),\label{eq:firm_preferences}
\end{equation}
where $\tilde{g}=\mathbf{E}g(h,X_{i})$ and the expectation is taken
with respect to the distribution of $X_{i}$.\footnote{Here, we implicitly assume that firms do not observe workers' covariates and rank workers only in terms of their education. We make these assumptions concerning firm information explicit in a subsequent section.} Suppose that $X_{i}$
is iid, $K$ takes two values $k_{1},$ $k_{2}$ and there are two
levels of education, $h_{1}$, $h_{2}$ with $h_{2}>h_{1}$. Let us
denote the set of firms that prefer high education ($h_{2}$) as 
\begin{align*}
M_{2}^{+}(\theta) & =\{m\in\{1,2\}:\rho(k_{m},h_{2};\theta)\geq\rho(k_{m},h_{1};\theta)\}.
\end{align*}

If $f(h,k)$ is of the form $a(h)+b(k)$, where $a$ and $b$ are
two functions that map the capital variables to the real numbers,
then $M_{2}^{+}(\theta)$ will be either $\{1,2\}$ or $\emptyset$.
In this case we say that firms have homogeneous preferences, since
both types of firms in the economy prefer the higher educated workers.
Alternatively, if $f(h,k)$ is of the form $a(h)b(k)$ then $M_{2}^{+}(\theta)$
will be either $\{1,2\}$, $\emptyset$, or $\{2\}$. This is the
case of heterogeneous firm preferences. In this latter case where
$f$ exhibits complementarities in worker and firm types, the set
of firms types that prefer high to low education is more finely partitioned.
Moreover, the presence or absence of complementarities will play a
key role in determining the severity of wage inequality. More general
than all these points, however, is the following fact about the model:
as long as $k$ appears somewhere in $f$, $k$ does not have to interact
directly with $h$ in $f$ for the information frictions represented
by $\beta$ to matter in worker's investment decision. 

\subsection{Some Implications of Frictional Matching Model\label{subsec:Some-Implications-Of}}

In this section, we explore some key features of the model. We will
suppose that the functional forms, underlying distributions, and firm
preferences are such that firms always strictly prefer higher educated
workers. In the following subsection, we will illustrate sorting without
any direct interactions between worker and firm types in the production
function. 

\subsubsection{Sorting Without Complementarities}

In Figure 1 and Figure 2, we compare the equilibrium probability of
investing in education and the equilibrium Gini coefficient for a
range of the friction parameters under two specifications of the production
function: Specification 1 allows direct interaction between worker
and firm types, $f=\theta_{1}hk$, while such interactions are absent
in Specification 2, $f=\theta_{1}(h+k)$. Each point on the plot is the average of 500 simulations of endogenous
variable from the equilibrium of the model. The outside option parameter
is set to $\theta_{2}=(-.75,.25,.5)$. There are 500 workers and firm positions.
In Specification 1, the high value of $\theta_{1}$ is 3, and the
low value of $\theta_{1}$ is 1. In Specification 2, the high value
of $\theta_{1}$ is 2, and the low value of $\theta_{1}$ is 1.
There are two levels of of firm capital: $K=1/2$ and $K=1$. The
fraction of each type of firm is .5 in the economy.\footnote{Across all the specifications, we set the  outside option function to be $g =\exp(h\cdot x\theta_{2})$.}

A number of implications are straightforward: the equilibrium probability
of investing in high education is higher when $\theta_{1}$ is higher
and frictions are lower. When $\theta_{1}$ is higher, workers will
be compensated more for higher levels of education. When $\beta$
is higher, the probability of matching to a higher type firm when
they choose high education is higher. 

The effect of increasing $\beta$ (lowering matching frictions) on
both the education and wage inequality is typically much more dramatic
in Specification 1. A rise in $\beta$ (a lessening in matching frictions)
increases sorting in both specifications. In Figure 1, the correlation
between worker and firm types rises from approximately zero to 45\% when $\theta_{1}$
is high, but from zero to 70\% when $\theta_{1}$ is lower; in Figure
2, the correlation between worker and firm types rises from zero to
51\% in the high theta case whereas it rises from zero to 68\% in
the low theta case. The overall level of inequality in Specification
1 is also higher since whatever sorting is present is amplified to
a greater extent when the types interact in the wage equation than
when they do not. 

The high $\theta_{1}$ case in the right hand panel of Figure 1 also
illustrates the role that two competing effects of changes in $\beta$
play on the level of wage inequality. When $\beta$ rises from 0 to
1, the level of inequality increases through the sorting channel.
However, as $\beta$ continues rises, the equilibrium probability
of investing in education also continues to rise. As the fraction
of highly educated surpasses 80\%, the level of inequality begins
to level off (at $\beta=2$) and then begins to fall. This phenomenon
is also illustrated to a lesser degree in the high $\theta_{1}$ case
of the right hand side panel of Figure 2. 

\label{SWC1}
\begin{figure}[H]
\caption{Education and Wage Inequality Under Specification 1}
\begin{centering}
\includegraphics[scale=0.30]{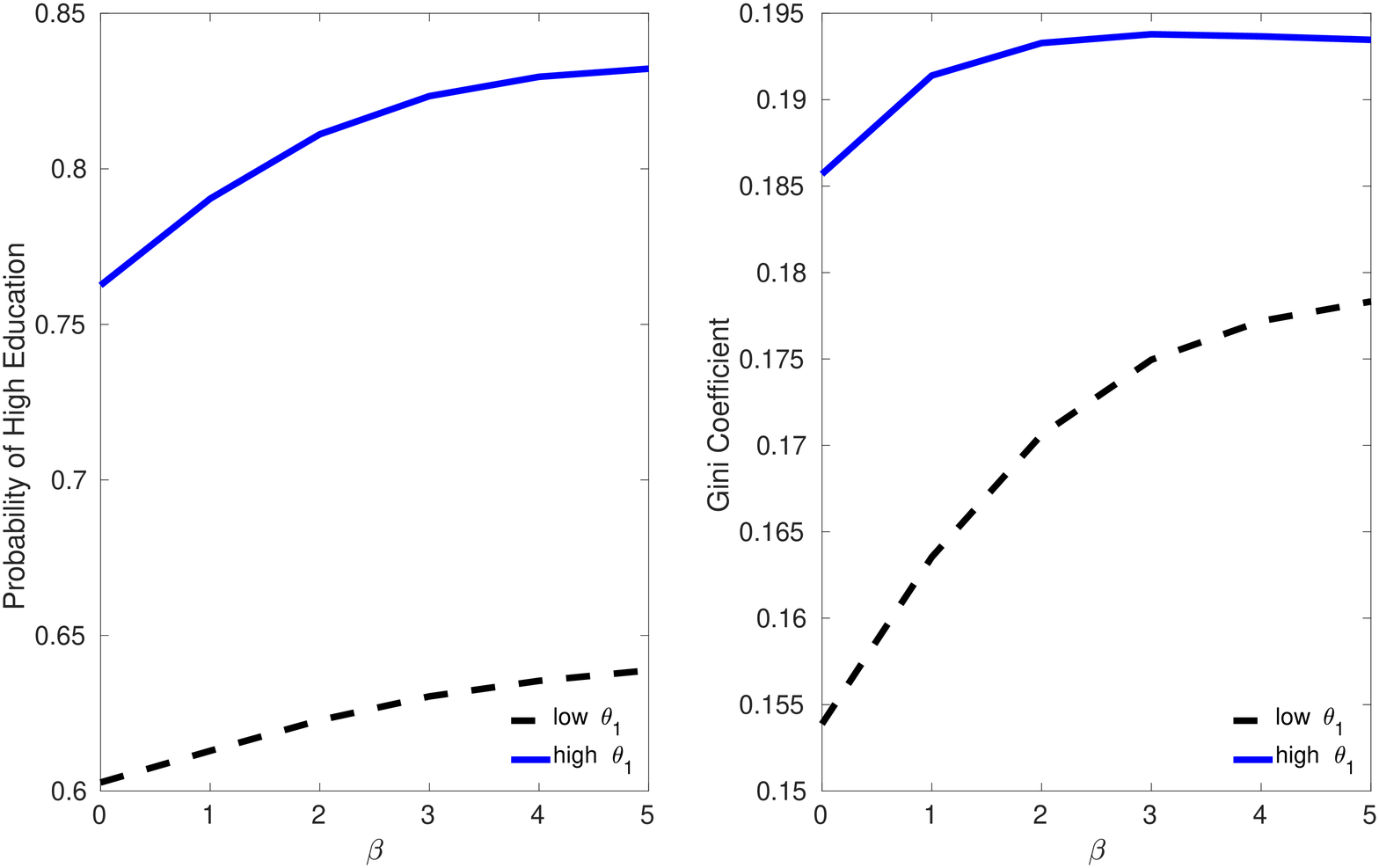}
\par\end{centering}
\caption{Education and Wage Inequality Under Specification 2}
\begin{centering}
\includegraphics[scale=0.30]{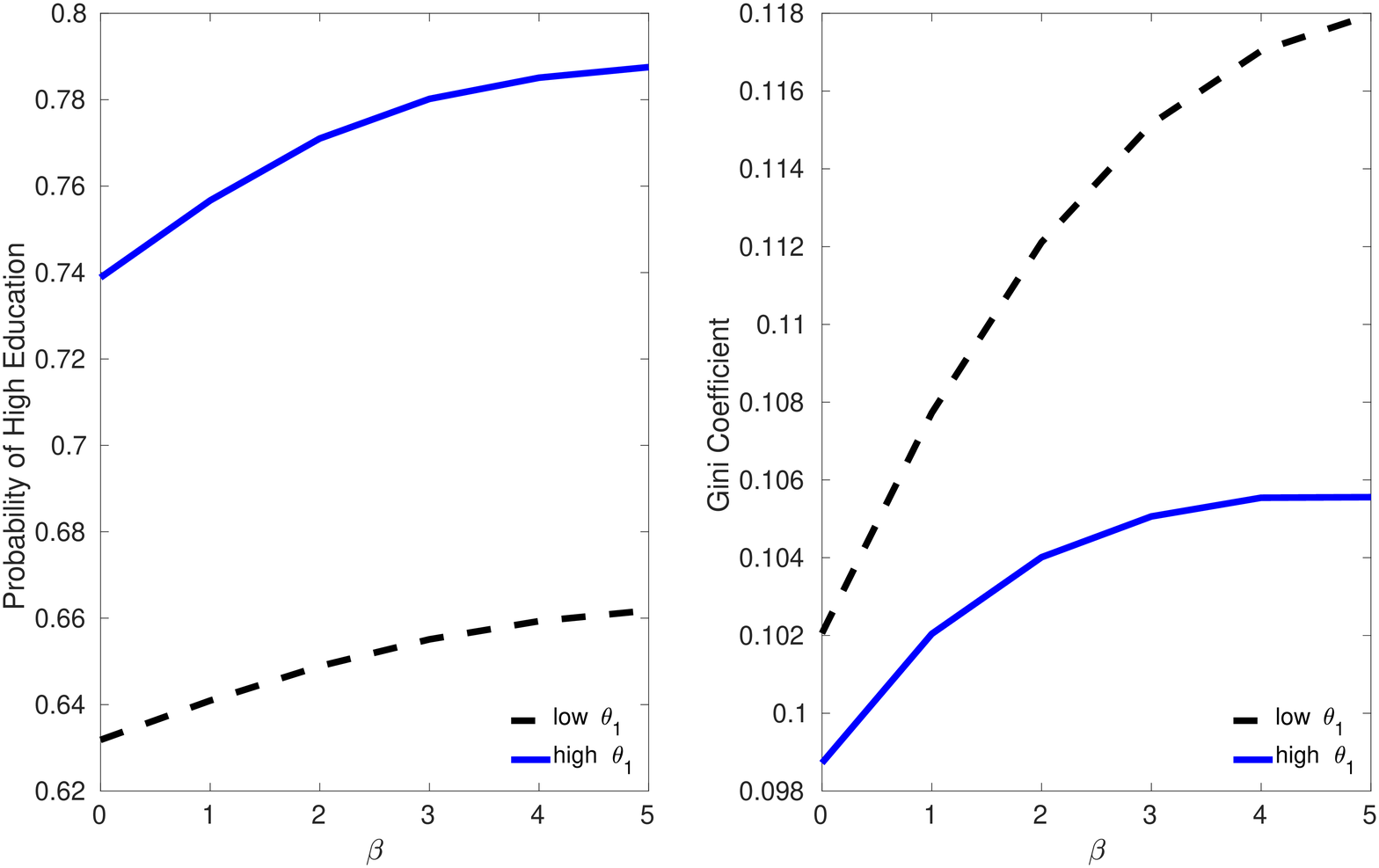}
\par\end{centering}
\parbox{6in}{\footnotesize Figures 1 and 2 plot the equilibrium
probability of high education investment and the Gini coefficient
for a range of values of the matching frictions parameter, $\beta$ in cases where firms all prefer higher-educated workers.
We consider two specifications for the production function: Specification
1 includes interactions between worker and firm types while Specification
2 does not. Lowering matching frictions (increasing $\beta$) increases
the equilibrium level of education across specifications. A rise in
$\beta$ impacts inequality through two competing effects: a sorting
effect that increases inequality and an a supply effect that lowers
inequality. This can be seen most dramatically in Figure 1: as $\beta$
rises past a value of three, the fraction of highly educated rises
more and more and inequality falls, dominating the effects of sorting
on inequality. }
\end{figure}

\subsubsection{Supply of Highly Educated Workers and Education Premia}

In this section, we show how simulation of our static model can capture a puzzling
phenomenon discussed in \citet*{Card/Lemieux:01:QJE}. How can dramatic increases in the
education wage premium lead to only modest increases in the supply
of highly educated workers? The authors note that, over a roughly
30 year period beginning in the early 1970s, the college-high school
wage gap rose considerably in the United States, Canada, and the United
Kingdom, and that this rise occurred mostly for younger workers. They
argue that an important source of this trend is a stagnation in the
rate of educational attainment among workers born in the 1950s and
thereafter.

In Figure 3, we show how this pattern can be driven entirely by changes
in the matching technology over time. The wage premium is measured
as the difference between the average wages of the workers with high
education and the average wages of workers with low education. Each
point on the plot represents the average of 500 simulations of the
model. We use Specification 1, $f=\theta_{1}hk$, under the same setup
as before with only one difference; we choose the low value of $\theta_{1}$
to be 0.7 and the high $\theta_{1}$ to be 2.5. In the case that $\theta_{1}$
is very low, the effect of raising $\beta$ is to dramatically increase
sorting without inducing a large benefit to the workers from acquiring education in equilibrium.

\begin{figure}[H]
\caption{Supply of Highly Educated Workers and Education Wage Premia}
\noindent \begin{centering}
\includegraphics[scale=0.30]{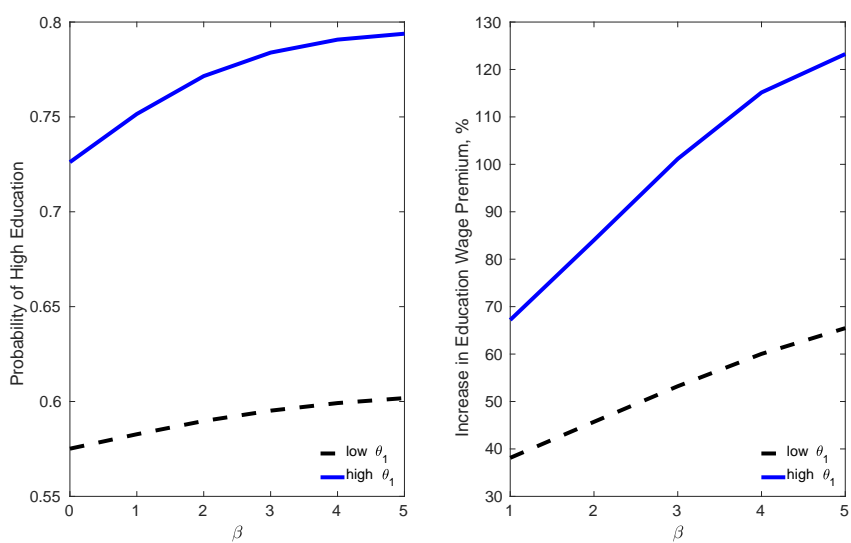}
\par\end{centering}
	\parbox{6in}{\footnotesize Figures 3 offers an explanation to
an empirical puzzle discussed in \citet*{Card/Lemieux:01:QJE}: why are increases
in wage premia not associated with large increases in the supply of
highly educated workers? We plot the equilibrium probability of high
education investment and the returns to education for a range of values
of the matching frictions parameter, $\beta$ when firms prefer higher education. In the case that $\theta_{1}$ is very low, the effect of increasing
$\beta$ is to dramatically increase sorting while keeping the equilibrium returns
to education for any particular worker reasonably low. }
\end{figure}

\section{Econometric Inference\label{sec:section_3_econometric_inference}}

In this section, we outline the general empirical strategy for performing
inference on the underlying model parameters. In Section \ref{sec:section_3_econometric_inference},
we describe how the main model can be used to characterize the observed
distribution of the matching of workers to firm and hence the wages
of all the workers in the economy. The goal is to then use these representations
to construct confidence regions for the preference and matching technology
parameters.

However, if the model is high dimensional, the Monte Carlo inference
approach may be cumbersome to apply in practice. For this reason,
we propose a two-stage inference approach that relies on the construction
of a first-stage confidence interval for a subset of the model parameters.
We demonstrate this approach in practice in Section \ref{subsec:3_2_First-stage-estimation-of_theta}
by estimating the Bayesian game from \ref{sec:StructuralModel}
for fixed values of $\beta$.

\subsection{Two-Stage Inference Accommodating Cross-Sectional Dependence of Observed
Matching\label{subsec:section_3_subsection_1_two_stage_inference}}

The econometrician observes a matching of workers to firms, $\mathbf{M}=(\mathbf{M}(i))_{i\in N_{h}}$,
where for each $i\in N_{h}$, $\mathbf{M}(i)$ takes values in the
set of firms.\footnote{Throughout this paper, we will suppose that the matching is one-to-one
between workers and firms. In practice, ``firms'' in this context
can be viewed as positions at particular firms. } The main challenge associated with inference is the fact that the
distribution of $\mathbf{M}$ exhibits cross-sectional dependence
of a complicated form.  The matching
of workers to firms can be thought of as discrete choice problem on
the part of the firm where the choice sets of firms are endogenously
constrained by the choices of firms with higher $v$-indices, which
depends on $\beta$, $\eta$ and $k$. Hence, the event that worker
$i$ matches to firm $j$ cannot be considered independent from the
event that a worker $i'\neq i$ matches to firm $j$. Also, the fact
that firm preferences may be heterogeneous means we cannot condition
on the $v$-index and firm preferences in a way to remove the cross-sectional
dependence as was done by \citet*{Diamond/Agarwal:17:QE}. 

The econometrician observes the vector $\mathbf{M}\in\mathbf{R}^{n_{h}}$,
which represents a matching of workers to firms. Given the serial
dictatorship matching process, the joint distribution of $\mathbf{M}$
is known up to a parameter. Let $\mathbf{K}=(\mathbf{K}(i))_{i\in N_{h}}$,
where $\mathbf{K}(i)=K_{\mathbf{M}(i)}$; i.e., the capital of the
firm matched to by worker $i$. 

Our model also implies that the finite sample distribution of wages,
$(\mathbf{W}(i))_{i\in N_{h}}$, is known up to a parameter. Under
Nash bargaining (and a specification of the post-match wage function
based off an equation such as \ref{eq:preferences}), we have for
each $i\in N_{h}$
\[
\mathbf{W}(i)=w(H_{i},\mathbf{K}(i)).
\]
We denote all the match-related observables as $\mathbf{Y}=(\mathbf{K},\mathbf{M})$.
$\mathbf{M}$ is observed whenever the researcher has matched employer-employee
data. $\mathbf{K}$ is observed when the researcher can use the matching
data, $\mathbf{M}$, and the firm capital data, $K$, to find the
capital of the firm each worker in the sample is employed at. Using
$\mathbf{Y}$ and worker observables $H$ and $X$, the econometrician
wishes to infer $\lambda_{0}$. 

\subsubsection{Finite Sample Inference on Parameters\label{subsec:sub_sub_section_finite-sample-inference}}

Next, we consider a test statistic that matches the moments of the
distribution of the matched-related observables with their simulated
counterparts. To simplify the exposition, we discuss the construction
of a confidence interval for $\beta_{0}$ alone, i.e., supposing that
we knew the true values of $\theta_{0}$. Denote $R+1$ as the total
number of simulations in the Monte Carlo inference procedure. Drawing
$\eta_{r}$ from some continuous parametric distribution function,\footnote{We will specify a particular parametric family that this distribution belongs to, along with additional assumptions, in Section \ref{subsec:3_2_First-stage-estimation-of_theta}. } we simulate a version of the matching for each $\beta\in B$ and
each $r=1,...,R+1$, which we write as $\mathbf{M}_{r}(\beta)=\{\mathbf{M}_{r}(i;\beta):i\in N_{h}\}$.
The simulated wages are then
\[
\mathbf{W}_{r}(i;\beta)=w(H_{i},K_{\mathbf{M}_{r}(i;\beta)})).
\]
It is convenient to define
\begin{eqnarray*}
\mathbf{Y}_{r}(\beta) & = & \{\mathbf{Y}_{r}(i;\beta):i\in N_{h}\},\\
\mathbf{Y}_{R+1}(\beta) & = & \{\mathbf{Y}_{r}(i;\beta):i\in N_{h},r=1,...,R+1\},\text{ and}\\
\mathbf{Y}_{-r}(\beta) & = & \mathbf{Y}_{R+1}(\beta)\backslash\mathbf{Y}_{r}(\beta).
\end{eqnarray*}
Next, we will propose a test statistic that depends on both the observed
matching data, $\mathbf{Y}$, and the simulated matching data, (along
with simulated versions of this test statistic). That is,
\begin{eqnarray*}
T(\beta)=\phi_{n}(\mathbf{Y},\mathbf{Y}_{R}(\beta)), & \text{ and}\\
T_{r}(\beta)=\phi_{n}(\mathbf{Y}_{r}(\beta),\mathbf{Y}_{-r}(\beta)).
\end{eqnarray*}
An example of such a test statistic is one that compares the observed
joint distribution of worker human capital and matched firm capital
with simulated counterparts. For example, we may consider the test
statistic\footnote{This test statistic is similar to the one used in \citet*{Kim/Schwartz/Song/Whang:19:Econometrics}. See also \citet{Diamond/Agarwal:17:QE}.}
\begin{align*}
T(\beta) & =\frac{1}{R}\sum_{r=1}^{R}\max_{h,m}|\hat{P}(h,m;\mathbf{Y},H)-\hat{P}(h,m;\mathbf{Y}_{r}(\beta),H)|,
\end{align*}
where \begin{align*}
\hat{P}(h,m;\mathbf{Y},H) & =\frac{1}{n_{h}}\sum_{i\in N_{h}}1\{H_{i}=h,K_{\mathbf{M}(i)}=m\}.
\end{align*}
\noindent  That is, $\hat{P}$ is an $J\times M$ matrix\footnote{In this example, we are implicitly assuming that the distribution
of $K$ is discrete and has $M$ support points. We will make this assumption explicit in a subsequent section.} whose $(j,m)$ element
is the estimated probability that a worker of education level $h_{j}$
matches to a firm of capital level $m$. $\hat{P}_{r}(\beta)$ is
defined similarly to $\hat{P}$, except we replace the observed matching
with the $r$th simulated matching, $\mathbf{M}_{r}(\beta)$.

Using our test statistic, we may compute a confidence region for
$\beta$ as 
\[
C_{\alpha,R}^{\beta}=\{\beta\in B:T(\beta)\leq c_{\alpha,R}(\beta)\},
\]
where the critical value is computed as the $(1-\alpha)$ -quantile
of the empirical distribution of $\{T_{r}(\beta):r=1,...,R\}$: 
\begin{eqnarray*}
c_{\alpha,R}(\beta) & = & \inf\left\{ c\in\mathbf{R}:\frac{1}{R}\sum_{r=1}^{R}1\{T_{r}(\beta)\leq c\}\geq1-\alpha\right\} .
\end{eqnarray*}
Under Assumption \ref{Assumption-4}, it can easily be shown that finite sample inference
on $\beta_{0}$ satisfies $P\{\beta_{0}\in C_{\alpha,R}^{\beta}\}\geq1-\alpha$
when the procedure outlined above involves the true parameter, $\theta_{0}$. 

In practice, we do not know the true value of $\theta_{0}$. In situations
in which the full parameter vector $\lambda_{0}$ is not very large,
it may be feasible to construct a $(1-\alpha)100\%$ confidence region
for this parameter that exhibits finite sample validity. That is,
we construct
\begin{equation}
C_{\alpha,R}^{\lambda}=\{\lambda\in\Lambda:T(\lambda)\leq c_{\alpha,R}(\lambda)\},\label{eq:CI_for_full_parameter_vector_finite_sample_validity}
\end{equation}
where $T(\lambda)$ and $c_{\alpha,R}(\lambda)$ are defined analogously
to $T(\beta)$ and $c_{\alpha,R}(\beta)$. In the case
that $\Lambda$ is high-dimensional, the finite sample
procedure outlined above may not be practical due to the unreasonable
computational cost. In the following subsection, we explore a two-stage
inference approach that admits inference on $\beta_{0}$ when the
researcher is able to construct a first-stage confidence region for
a subset of the parameters, $\theta_{0}$. 

Note that plugging in a consistent estimator of $\theta_{0}$,
$\hat{\theta}_{n}$, for the true value in inference procedure outlined
above will generally not lead to valid inference on $\beta_{0}$.
This is because there is no reason to expect that plugging in $\hat{\theta}_{n}$
for $\theta_{0}$ will make the distribution of the simulated matching,
$\mathbf{M}_{r}$, equal to the distribution of the observed matching,
$\mathbf{M}$. The fact that $\mathbf{M}_{r}$ is not equal in distribution
to $\mathbf{M}$, in turn implies that $\mathbf{K}_{r}$ does not
follow the same distribution as $\mathbf{K}$. The severe consequences
of estimation error in $\hat{\theta}_{n}$ occur because the firm
preferences are typically misspecified at all values of $\theta$
other than the true value, $\theta_{0}$. Moreover, this problem is
not alleviated by conditioning on $H,K,$ or exogenous variables.
In the following section, we discuss a general two-stage inference
approach that can be used when the econometrician can construct an (asymptotically)
valid confidence first-stage confidence interval for $\theta_{0}$.
In Section \ref{sec:StructuralModel},
we extend our baseline economic model of Section \ref{sec:section_two_structural_model}
in a manner that admits the application of this two-stage inference
approach to our setup.

\subsubsection{Two-Stage Inference on $\beta$ using Test-Inversion Confidence Interval }

Suppose that we wish a ($1-\alpha$)-level asymptotic confidence interval
for $\beta_{0}$, and can construct a confidence interval for $\theta_{0}$.
Let us denote the test statistic and its simulated counterpart from
the previous section, where the $\theta$ arguments make explicit
the test statistic's dependence upon a given value of $\theta\in\Theta$:
\begin{eqnarray*}
T(\beta;\theta_{0},\theta_{1})=\phi_{n}(\mathbf{Y}(\beta_{0},\theta_{0}),\mathbf{Y}_{R}(\beta,\theta_{1})), & \text{ and}\\
T_{r}(\beta;\tilde{\theta},\theta_{1})=\phi_{n}(\mathbf{Y}_{r}(\beta,\tilde{\theta}),\mathbf{Y}_{-r}(\beta,\theta_{1})).
\end{eqnarray*}
Note that according to the notation we used in the last section we
have $T(\beta;\theta_{0},\theta_{0})=T(\beta)$. 
Our inference on
$\beta$ proceeds in two steps:

\noindent \textbf{Step 1. }Using the first stage estimates of $\hat{\theta}(\beta)$,
we construct a confidence region for $\theta_{0}$, $\hat{C}_{\alpha/2}(\beta)$,
with $(1-(\alpha/2))$ asymptotic coverage.

\noindent \textbf{Step 2.} Next, we construct a test statistic that
doesn't involve $\theta$.\textbf{ }Define 
\begin{eqnarray*}
S(\beta) & = & \inf_{\theta_{1}\in\hat{C}_{\alpha/2}(\beta)}T(\beta;\theta_{0},\theta_{1}),\text{ and}\\
S_{r}^{*}(\beta) & = & \sup_{\tilde{\theta}\in\hat{C}_{\alpha/2}(\beta)}\inf_{\theta_{1}\in\hat{C}_{\alpha/2}(\beta)}T_{r}(\beta;\tilde{\theta},\theta_{1}).
\end{eqnarray*}
We now construct a confidence set for $\beta$ as
\begin{equation}
\hat{C}_{\alpha,R}=\{\beta\in B:S(\beta)\leq c_{1-(\alpha/2),R}^{*}(\beta)\},\label{eq:Cbeta_inversion}
\end{equation}
where the critical value $c_{1-(\alpha/2),R}^{*}(\beta)$ is computed
as the $(1-(\alpha/2))$ -quantile of the empirical distribution of
$\{S_{r}^{*}(\beta):r=1,...,R\}$; that is, 

\begin{eqnarray*}
c_{1-(\alpha/2),R}^{*}(\beta) & = & \inf\left\{ c\in\mathbf{R}:\frac{1}{R}\sum_{r=1}^{R}1\{S_{r}^{*}(\beta)\leq c\}\geq1-(\alpha/2)\right\} .
\end{eqnarray*}

The following lemma establishes the asymptotic validity of the two-stage
inference procedure. 
\begin{lemma}
Suppose that the econometrician can construct $\hat{C}_{\alpha/2}(\beta_{0})$
such that
\[
\lim_{n\rightarrow\infty}P\left(\theta_{0}\in\hat{C}_{\alpha/2}(\beta_{0})\right)\geq1-(\alpha/2).
\]
Then
\begin{equation}
\lim_{n\rightarrow\infty}P\left(\beta_{0}\in\hat{C}_{\alpha,R}\right)\geq1-\alpha.\label{eq:Section_3_lemma_1}
\end{equation}
\end{lemma}
\begin{proof}

By the definition of $\hat{C}_{\alpha,R}$, $P\left(\beta_{0}\in\hat{C}_{\alpha,R}\right)$ is equal to 
\noindent 
\begin{eqnarray}
P\left(S(\beta_{0})\leq c_{1-(\alpha/2),R}^{*}(\beta_{0})\right) & = & P\left(\inf_{\theta_{1}\in\hat{\mathcal{C}}_{\alpha/2}(\beta_{0})}T(\beta_{0};\theta_{0},\theta_{1})\leq c_{1-(\alpha/2),R}^{*}(\beta_{0})\right)\nonumber \\
 & \geq & P\left[\left\{ \inf_{\theta_{1}\in\hat{\mathcal{C}}_{\alpha/2}(\beta_{0})}T(\beta_{0};\theta_{0},\theta_{1})\leq c_{1-(\alpha/2),R}^{*}(\beta_{0})\right\} \cap A_{1}\right],\label{eq:-1}
\end{eqnarray}

\noindent where $A_{1}\equiv\left\{ \theta_{0}\in\hat{C}_{\alpha/2}(\beta_{0})\right\} $.
Then, the right hand side of the right hand side of (\ref{eq:-1}) is greater than or equal
to 
\noindent 
\begin{eqnarray*}
 & & P\left[\left\{\sup_{\tilde{\theta}\in\hat{C}_{\alpha/2}(\beta_0)}\inf_{\theta_{1}\in\hat{\mathcal{C}}_{\alpha_{/2}}(\beta_0)}T_{r}(\beta_0;\tilde{\theta},\theta_{1})\leq c_{1-(\alpha/2),R}^{*}(\beta_{0})\right\}\cap A_{1}\right],\\
 & \geq & P\left(\sup_{\tilde{\theta}\in\hat{C}_{\alpha/2}(\beta_0)}\inf_{\theta_{1}\in\hat{\mathcal{C}}_{\alpha_{/2}}(\beta_0)}T_{r}(\beta_0;\tilde{\theta},\theta_{1})\leq c_{1-(\alpha/2),R}^{*}(\beta_{0})\right)-P\left(A_{1}^{c}\right).
\end{eqnarray*}
 Now since
\begin{eqnarray*}
\lim_{n\rightarrow\infty}P\left(\theta_{0}\notin\hat{C}_{\alpha/2}(\beta_{0})\right) & \leq & \alpha/2,
\end{eqnarray*}
we have
\begin{eqnarray*}
\lim_{n\rightarrow\infty}P\left(\beta_{0}\in\hat{C}_{\alpha,R}\right) & \geq & 1-\alpha.
\end{eqnarray*}
\end{proof}
\noindent

In the following section, we provide assumptions under
which we can construct a confidence
region for $\theta_{0}$ using a maximum likelihood approach. In Section $\ref{sec:identification}$,
we argue that $\theta_{0}$ is identified up to $\beta_{0}$, and
provide standard conditions under which the maximum likelihood estimator
is consistent and asymptotically normal. In Section $\ref{sec:MonteCarlo}$,
we then present a small Monte Carlo study that illustrates how this
estimator can be used as the basis for the first-stage inference on
preferences. In particular, we show how a parametric bootstrap can
be used to construct a $\hat{C}_{\alpha,R}$ with reasonable finite
sample size and power properties.

\subsection{First-Stage Estimation of Preferences \label{subsec:3_2_First-stage-estimation-of_theta}}

In this section, we show how $\theta$ can be estimated for a particular
fixed value of $\beta$. We will write an estimator of such an object
as $\hat{\theta}(\beta)$. The main challenge associated with  this
problem is that of estimating the worker's expected utility from equation
\ref{eq:expected_utility_fn}. The problem is difficult because the
workers must somehow resolve uncertainty associated with the serial
dictatorship matching process in order to compute the expected output
under the equilibrium education choices. In spite of these complications,
it turns out that, under reasonable assumptions, the parameters are
tractably estimable using discrete choice methods with a fixed point
constraint when there are only two education choices. We now provide
and discuss these assumptions. 
\noindent \begin{assumption}\label{AssInfo} 
\noindent 
\textit{(a) Firms observe (i) workers' education decisions,
$H$, and (ii) the distribution of characteristics, $X$. (b) Workers observe (i) the distribution of firm
capital, (ii) the distribution of $\eta$, (iii) the distribution
of $X$, and (iv) the distribution of the number of firms preferring
each education level $h_{j}\in\mathcal{H}$.}
\end{assumption}

Under part (a) of Assumption \ref{AssInfo}, firms do not take covariates into
account when forming their preference rankings over workers. Thus,
workers with the same education level are equally desirable to any
given firm. When worker $i$ considers the desirability of choosing
education $h_{j}$, he need only consider the capital a generic agent
who chooses level $h_{j}$ expects to receive in the matching process.
In many contexts, (a) will be reasonable for a host of variables that
affects the worker's education decision (e.g., marital status, number
of dependent children).\footnote{In some cases in which employers do see these worker characteristics,
they are prohibited from discriminating based on them due to state
or federal anti-discrimination laws.} Part (b) says that workers know only the distribution of firm capital
without knowing the precise realizations of capital. Assumption \ref{AssInfo}
(b) also stresses that the worker's knowledge of the distribution
of capital is not sufficient for knowledge of the distribution of
the number of firms that prefers each education class, which will
turn out to be crucial for our results of this section. 

\noindent \begin{assumption}\label{Assumption-4}

\noindent \textit{(a) $K$ is discrete with probability mass $q=(q_{m})_{m=1}^{M}$, where for $m=1,...,M$,  $q_{m}=P(K=k_{m})$. }

\noindent \textit{(b) $\eta_{j}$'s are iid $N(0,\sigma^{2})$ }
\noindent \textit{(c) $\varepsilon_{i}$'s follow the Type I extreme
value distribution.}
\end{assumption}

Part (a) says the distribution of firm capital has discrete support.
In practice, we can let $M$ be as large as our application requires.
In concert with (b) and the parametric structure for $v$ stipulated
by equation \ref{eq:section_two_v_functional_form}, (a) allows us
to express the unconditional distribution of $v_{j}$ as a mixture
of normals, $G\equiv\sum_{m=1}^{M}q_{m}F_{m}$, where $F_{m}$ is
$N(\beta k_{m},\sigma^{2})$.\footnote{In the simulation sections of the paper we normalize
$\sigma^{2}=1$ when we perform inference on the model parameters.} Part (c) is an assumption on the worker's unobserved costs that allows
us to estimate the model parameters using conventional discrete choice
methods. 

We wish to obtain a convenient representation of each worker's conditional
expectation of the production function, for each education level that
the worker can choose. Under the model of Section \ref{sec:StructuralModel}
the identity of the firm that worker $i$ matches with, $\mathcal{M}(i)$,
depends on $K$, $H$, $\beta$, and, $\theta$. Therefore, for each
$i\in N_{h}$ and $h_{j}\in\mathcal{H}$, we wish to estimate 
\begin{align*}
\tilde{f}_{ij} & \equiv\mathbf{E}[f(H_{i},K_{\mathcal{M}(i)})|H_{i}=h_{j},X_{i}=x_{i}],
\end{align*}
where the expectation is taken with respect to the distribution of
$K$, $H_{-i}$ and $\eta$. 
Under Assumption \ref{Assumption-4} (a), we can express the expectation on the preceding
line as 
\begin{equation}
\tilde{f_{ij}}=f_{j}'\pi_{j}^{(i)},\label{eq:f_ij_tilde}
\end{equation}
where $f_{j}=(f_{j1},...,f_{jm})'$ is an $M\times1$ vector with
the $m$-th element of $f_{j}$ given as $f_{jm}=f(h_{j},k_{m})$ and
$\pi_{j}^{(i)}=(\pi_{1j}^{(i)},...,\pi_{Mj}^{(i)})'$ is an $M\times1$
vector with the $m$-th element of $\pi_{j}^{(i)}$ given as 
\begin{eqnarray}
\pi_{mj}^{(i)} & = & \sum_{h_{-i}\in\mathcal{H}_{-i}}P(\mathcal{M}(i)=m|H_{i}=h_{j},H_{-i}=h_{-i},X_{i}=x_{i})P(h_{-i}|x_{i}).\label{eq:section_3_pi_m_j_i}
\end{eqnarray}
This is the probability that worker $i$ matches to a firm of capital
level $k_{m}$ when he has chosen education level $h_{j}$.\footnote{Note that although these terms depend on $\theta$ and $\beta$, we
will occasionally omit these from our notation for convenience.} Given that there are $M$ education levels, $J$ choices, and $n_{h}$
workers, the dimensionality of the problem appears daunting. However,
under our assumptions the problem is simplified considerably, and
we can show that for each $j$ and $m$, $\pi_{mj}^{(i)}=\pi_{mj}$, and hence, $\tilde{f}_{ij}=\tilde{f}_{j}$.\footnote{The argument for why this is the case is given in the proof of Proposition
\ref{Proposition_1}.}

Although it is unclear how to represent $\pi_{mj}$'s analytically when
the worker faces a choice between a large number of education levels,
the problem becomes tractable when there are only two (i.e., $J=2$).
Proposition \ref{Proposition_1} shows that under our informational
assumptions, firms (and workers) cannot distinguish between workers
with the same education level during the matching process. As a consequence, we find that a worker is only concerned with the number of
other workers who picked one of the two education levels (and not
which particular workers chose what). Independence and identical distributions
assumptions imply that the probability that $n_{j}$ workers picked
education level $h_{j}$ can be represented using the binomial probability
mass function. However, the number of workers choosing education level
$h_{j}$ is unknown to workers, so they must take expectations. Thus,
instead of having to sum over $n_{h}-1$ indices associated with actions
of each of the other workers to compute the worker's expectation,
we need only sum over one: the number of workers choosing a particular
education level.

 We will also allow $\theta$ to enter
$\pi_{mj}$'s through the distribution of the number of firms that prefer
high (or low) education. The following assumption is a natural way
to specify this distribution. We use the notation $M_{j}^{+}(\theta)$ to denote the set of firm types that prefer education level $h_{j}$.\footnote{That is, $M_{j}^{+}(\theta)  =\{m\in\{1,...,M\}:\rho(k_{m},h_{j};\theta)\geq\rho(k_{m},h_{j'};\theta),j\neq j'\}$. See also the discussion before Proposition    \ref{Proposition_2_heterogeneous_firms}.}
\noindent \begin{assumption}\label{4b}\textit{In the model with $J=2$,
the probability that exactly $n^{(j)}$ firms prefer workers with
education level $h_{j}$ follows the binomial distribution with probability
$\sum_{m\in M_{j}^{+}(\theta)}q_{m}$.}
\end{assumption}

The explicit representation of the matching probabilities are given
in Propositions \ref{Proposition_2_heterogeneous_firms}, \ref{Proposition_3_homogeneous_firms}, and Lemma \ref{Lemma_4_order_statistic_lemma}. These results can be used to construct estimates of the $\pi_{mj}$'s
- and hence the $\tilde{f}_{j}$'s - for fixed values of $\theta$
and $\beta$.  Using a given functional form for the production function,
we denote an estimate of the expected production function when the
worker chooses education level $h_{j}$ as
\begin{align*}
\hat{f}_{j}(\theta,\beta) & =f_{j}'\hat{\pi}_{j}(\theta,\beta),
\end{align*}
where our notation emphasizes the dependence of the objects upon the
parameter values. To construct $\hat{\pi}_{mj}$'s we must estimate the terms of equation
\ref{eq:pimjtheta}. $\hat{P}(n_{j})$ is constructed as $B(n_{j};n_{h}-1,\hat{p}_{j})$
where the latter denotes the binomial probability mass function with $\hat{p}_{j}=P(H_{i}=h_{j})$.\footnote{In so doing, we pursue a two-step approach for estimating the choice probabilities, such as \citet*{bajari2013game}. See for example \citet*{Kasahara/Shimotsu:12:Ecta} for an alternative approach.}
Similarly, $\hat{P}(n^{(j)};\theta)$ is constructed as $B(n_{j};n_{h}-1,\hat{q}_{j}(\theta))$,
where $\hat{q}_{j}(\theta)=\sum_{m\in\hat{M}_{j}^{+}(\theta)}\hat{q}_{m}$, with
$\hat{q}_{m}=\hat{P}(K_{j}=m)$,  
\[
\hat{M}_{j}^{+}(\hat{\theta})=\left\{ m\in\{1,...,M\}:\hat{\rho}(k_{m},h_{j};\theta)\geq\hat{\rho}(k_{m},h_{j'};\theta),j\neq j'\right\} ,
\]
and $\hat{\rho}(k_{m},h_{j};\theta)$ is as in equation  (\ref{eq:firm_preferences}),
except we use $\hat{g}_{j}=\frac{1}{n}\sum_{i=1}^{n}g(h_{j},X_{i})$
in place of $\tilde{g}$. 

Lastly, the $P_{h_{j},n_{j},n^{(j)}}(m)$'s, from equation \ref{eq:pimjtheta}
- that is, the probability that a worker matches to a firm of type
$m$ when they choose education level $h_{j}$, $n_{j}$ other workers
choose $h_{j}$, and $n^{(j)}$ firms prefer $h_{j}$ - can be simulated
for fixed values of $\theta$ and $\beta$. Propositions \ref{Proposition_2_heterogeneous_firms}
and \ref{Proposition_3_homogeneous_firms} show how these
can be represented using probabilities involving order statistics.
Under Assumption \ref{Assumption-4} (b), we can construct $\hat{P}_{h_{j},n_{j},n^{(j)}}(m)$'s
by averaging functions of simulated draws of beta-distributed random variables.\footnote{In particular,
see Corollaries \ref{corollary:1} and \ref{corollary:2}, which follow the order statistic
result in Lemma \ref{Lemma_4_order_statistic_lemma}.} 

Once  we have estimated $\hat{f}_{j}(\theta,\beta)$ for each education level, we may use the specification of the wage from equation \ref{eq:expected_wage} to write the expected wage as 
\begin{align}
\hat{\omega}_{ji}(\theta,\beta)= & \tau\hat{f}_{j}(\theta,\beta)+(1-\tau)g(H_{i},X_{i};\theta). \label{eq:expected_age}
\end{align}
When there are two choices ($J=2$), the worker chooses high education
($h_{j}=1$) if and only if
\[
U_{1i}^{*}-U_{i0}^{*}>0.
\]
 Under the assumption that $\varepsilon_{i}$'s follow the extreme
value distribution (Assumption \ref{Assumption-4}), the probability that worker $i$
chooses high education can be written as
\[
\hat{p}_{i}(\theta,\beta)=\frac{\exp(\hat{\omega}_{1i}(\theta,\beta)-\hat{\omega}_{0i}(\theta,\beta))}{1+\exp(\hat{\omega}_{1i}(\theta,\beta)-\hat{\omega}_{0i}(\theta,\beta))}.
\]
Since the covariates $\{X_{i}\}_{i=1}^{n}$
are iid we can write the joint likelihood as the product of the marginal
likelihoods. We can then define the estimator
of $\theta$ (for a fixed value of $\beta$) as the minimizer of the standard logit likelihood function:\footnote{When $\beta$ is fixed, maximizing the likelihood by computing the
$\hat{f}_{j}(\theta,\beta)$'s for each candidate value of $\theta$
can be slow. The following strategy can be used to estimate $\theta$
for fixed $\beta$ more quickly provided that the support of $K$
is not too large. First, note that $\theta$ enters $\hat{f}_{j}(\theta,\beta)$
only through the set of firm types that prefer education level $h_{j}$,
$\hat{M}_{j}^{+}(\theta)$. Given our assumptions on the production
function and firm preferences, $\hat{M}_{j}^{+}(\theta)$ must take
one of $M+1$ possible values. Therefore, for fixed $\beta$ , we
can avoid simulating $\hat{f}_{j}(\theta,\beta)$ for each candidate
value of $\theta$ by pre-allocating the $\hat{q}_{j}(\theta)$'s
and $P_{h_{j},n_{j},n^{(j)}}(m)$'s for each of the $M+1$ cases for
$\hat{M}_{j}^{+}(\theta)$. It then suffices to evaluate $\hat{M}_{j}^{+}(\theta)$,
select the appropriate dimension of the array of terms, then assemble
the terms according to equation \ref{eq:pimjtheta}.}

\[
\ln L_{n}(\theta,\beta)=-\sum_{i=1}^{n}\left[h_{i}\ln\hat{p}_{i}(\theta,\beta)+(1-h_{i})\ln(1-\hat{p}_{i}(\theta,\beta))\right].
\]

\subsection{Matching Probabilities }

In this section, we consider the role of frictions, or the magnitude
of $\beta_{0}$ relative to the variance of $\eta$, in shaping matching
patterns between workers and firms. Note that these frictions play
no role in determing firm preferences, or which firm types prefer
high education.\footnote{We discuss the role of firm preferences on matching patterns at the
end of Section \ref{sec:StructuralModel}.} Nevertheless, because the frictions do affect sorting patterns, they
are of considerable importance to workers when they decide how much
to invest in education. 

In the following example, we will suppose that that the set of firms
that prefer education level $h_{j}$, $M_{j}^{+}$, contains at least
two types of firms, $m$ and $\tilde{m}$ with $k_{m}\neq k_{\tilde{m}}$. Suppose we fix $N_{j}$, the number of workers who chose education
level $h_{j}$, at some $n_{j}$ and we fix $N^{(j)}$, the number of
firms who prefer highly-educated workers at some $n^{(j)}$ such that
$n_{j}+1<n^{(j)}$. In this situation, there are strictly more firms
who prefer type $h_{j}$ workers than there are workers of this type.
Let $\kappa=n^{(j)}-n_{j}+1$, and denote $p_{m\kappa}\equiv P(v_{m}>v_{(\kappa)})$
for each $m$ in $M_{j}^{+}$. Proposition \ref{Proposition_2_heterogeneous_firms} says that the difference
in the probability of matching to a type $\tilde{m}$ versus a type
$m$ firm at these values of $n_{j}$ and $n^{(j)}$ in such a situation
is given by
\begin{eqnarray}
 &  & \left(p_{\tilde{m}\kappa}-p_{m\kappa}\right)q_{\tilde{m}}^{+}/c_{\kappa}+p_{m\kappa}\left(q_{\tilde{m}}^{+}-q_{m}^{+}\right)/c_{\kappa},\label{eq:MP}
\end{eqnarray}
with
\[
c_{\kappa}\equiv\sum_{m\in M_{j}^{+}}p_{m\kappa}q_{m}^{+},
\]
where $q_{m}^{+}=q_{m}/\sum_{m\in M_{j}^{+}}q_{m}$. Under Assumption \ref{Assumption-4}, the case of $\beta_{0}=0$ gives us that $p_{mk}=p_{\tilde{m}k}$,
implying that the first term in the parentheses of equation \ref{eq:MP}
is zero. This means that when matching frictions are highest (i.e.,
when $\beta_{0}=0$), the difference in the probability of matching to
one type of firm that prefers $h_{j}$ over another is captured by the
relative prevalence of those types of firms in the economy. 

In the case that $\beta_{0}>0$, Assumption \ref{Assumption-4} implies that $p_{\tilde{m}\kappa}-p_{m\kappa}$
becomes larger as $k_{\tilde{m}}-k_{m}$ becomes larger. This means
that higher capital firms have a better chance of matching with the
high education workers when $\beta_{0}>0.$ On the other hand, in the
case that $n_{j}+1>n^{(j)}$ (i.e., $h_{j}$ is demanded by fewer firms
than there are in the economy), then the above probabilities once again depend solely on the relative
prevalence of the each type of firm.

\section{A Small Monte Carlo Simulation Study\label{sec:MonteCarlo}}
\noindent \begin{center}
\begin{table}[!htbp]
\caption{The Empirical Coverage Probability of Parametric Bootstrap Confidence
Intervals for $a'\theta_{0}$ at 95\% Nominal Level When $\beta_{0}$
is Known. }

\begin{centering}
\begin{tabular}{cccccc}
 &  & \multicolumn{4}{c}{}\tabularnewline
\hline 
 &  & \multicolumn{4}{c}{Specification}\tabularnewline
\cline{1-1} \cline{3-6} 
$\beta_{0}$  &  & $g_{1}$, $f_{1}$  & $g_{1}$, $f_{2}$ & $g_{2}$, $f_{1}$ & $g_{2}$, $f_{2}$\tabularnewline
\hline 
$0$  & $n=500$  & 0.9540  & 0.9600  & 0.9640 & 0.9720\tabularnewline
 & $n=1000$  & 0.9480 & 0.9480 & 0.9480 & 0.9480\tabularnewline
\hline 
$1$  & $n=500$  & 0.9360  & 0.9560  & 0.9740 & 0.9680\tabularnewline
 & $n=1000$  & 0.9400 & 0.9580 & 0.9480 & 0.9420\tabularnewline
\hline 
$2$  & $n=500$  & 0.9320  & 0.9700  & 0.9740 & 0.9660\tabularnewline
 & $n=1000$  & 0.9360 & 0.9660 & 0.9460 & 0.9460\tabularnewline
\hline 
$3$  & $n=500$  & 0.9360  & 0.9660  & 0.9740 & 0.9620\tabularnewline
 & $n=1000$  & 0.9460 & 0.9740 & 0.9400 & 0.9520\tabularnewline
\hline 
 & \multicolumn{1}{c}{} &  &  &  & \tabularnewline
\end{tabular}
\par\end{centering}
\begin{singlespace}
\parbox[c]{6in}{%
Notes: The table reports the empirical coverage probability of the
parametric bootstrap confidence interval for $a'\theta_{0}$, where
$a=(1,1)'$ and $\theta_{0}=(1,1)'$. The simulated rejection probability
at the true parameter is close to the nominal size of $\alpha=0.05$.
The simulation number is $R=500$. In each of the iterations, the
bootstrap number is $B=200$. %
}
\end{singlespace}
\end{table}
\par\end{center}

\noindent \begin{center}
\begin{table}[!htbp]
\caption{Average Length of Parametric Bootstrap Confidence Intervals for $a'\theta_{0}$
at 95\% Nominal Level When $\beta_{0}$ is Known.}

\begin{centering}
\begin{tabular}{cccccc}
 &  & \multicolumn{4}{c}{}\tabularnewline
\hline 
 &  & \multicolumn{4}{c}{Specification}\tabularnewline
\cline{1-1} \cline{3-6} 
$\beta_{0}$  &  & $g_{1}$, $f_{1}$  & $g_{1}$, $f_{2}$ & $g_{2}$, $f_{1}$ & $g_{2}$, $f_{2}$\tabularnewline
\hline 
$0$  & $n=500$  & 0.9186  & 1.0630  & 1.7494 & 1.1045\tabularnewline
 & $n=1000$  & 0.6509 & 0.7819 & 0.6826 & 0.7889\tabularnewline
\hline 
$1$  & $n=500$  & 0.8702  & 1.1872  & 1.3817 & 1.1953\tabularnewline
 & $n=1000$  & 0.6239 & 0.8676 & 0.6730 & 0.8581\tabularnewline
\hline 
$2$  & $n=500$  & 0.8733  & 1.2912  & 0.9566 & 1.2896\tabularnewline
 & $n=1000$  & 0.6320 & 0.9491 & 0.6588 & 0.9350\tabularnewline
\hline 
$3$  & $n=500$  & 0.8851  & 1.4045  & 0.9542 & 1.3866\tabularnewline
 & $n=1000$  & 0.6460 & 1.0349 & 0.6706 & 1.0110\tabularnewline
\hline 
 & \multicolumn{1}{c}{} &  &  &  & \tabularnewline
\end{tabular}
\par\end{centering}
\begin{singlespace}
\parbox[c]{6in}{%
Notes: This table reports the average length of the asymptotic confidence
interval for $a'\theta_{0}$, where $a=(1,1)'$ and $\theta_{0}=(1,1)'$.
The lengths of the of the confidence intervals decrease with $n$.
The simulation number is $R=500$ and the bootstrap number is $B=200$. %
} 
\end{singlespace}
\end{table}
\par\end{center}

In this section, we investigate the finite sample size and power properties
of the estimator of preferences, $\hat{\theta}_{n}(\beta)$, under
a variety of parameters and functional form assumptions. The results
in this section are for the case where the matching technology, $\beta$,
is known to the econometrician. 

We consider the following general structure for the worker's expected
utility function: 
\[
\tilde{U}_{i}=\left(f_{i}+g_{i}\right)/2+d(H_{i})\varepsilon_{i},
\]
where $\theta=(\theta_{1},\theta_{2})'\in\mathbf{R}^{2}$, $X_{i}\in\mathbf{R}^{2}$, and
$d(H_{i})$ is a $2\times1$ vector with one in the $H_{i}$-th row
where $H_{i}$ takes values of one or two. $X_{i}=(X_{1i},X_{2i})'$
are drawn independently across $i$ and one another from $U[0,1].$
Firm capital takes the value of $1/2$ and $1$ with equal probability.
We also suppose that $\varepsilon_{i}\in\mathbf{R}^{2}$ follows the
extreme value distribution so that the best response probability function
has the logit structure. $\eta$ is drawn independently from the
standard normal distribution. We use 100 draws of beta random variables
to compute the matching probabilities. We also interpolate the supports
of $N^{(j)}$ and $N_{(j)}$ so that they have $n/50$ support points
rather than $n-1$ support points. In our experiments, we set the
true value of preferences to be $\theta_{0}=(1,1)'$.

We consider two functional forms for the production function which
we call $f_{i1}$ and $f_{i2}$: 
\begin{align*}
f_{i1} & =\theta_{1}H_{i}\cdot\pi_{i}(\theta,\beta)'k,\text{ and}\\
f_{i2} & =\theta_{1}(H_{i}+\pi_{i}(\theta,\beta)'k).
\end{align*}
$f_{i1}$ implies production complementarities between the worker
and firm variables whereas any complementarities in $f_{i2}$ are
forced through the worker's expectation of firm capital $\pi_{i}'k$.
We also consider the performance of the inference under two versions
of the outside option, $g_{i1}$ and $g_{i2}$: 
\begin{align*}
g_{i1} & =\exp(H_{i}\cdot X_{i}\theta_{2}),\text{ and}\\
g_{i2} & =H_{i}\exp(X_{i}\theta_{2}).
\end{align*}
Note that these choices of the outside option function ensure non-negativity.
For each simulation sample, the $H_{i}$'s are generated as follows.
First, we solve for fixed point in the best response operator to obtain
$\mathbf{P}^{*}$.\footnote{In experiments with different starting values, iterating the best
response operator yielded the same fixed point each time.} Then we compute the best response at the simulated covariates

\[
\Psi_{i}(H_{i}|X_{i},\mathbf{P}_{-i}^{*})=\frac{\exp(\tilde{U}_{i}^{*}(H_{i},X_{i}))}{\sum_{j=1}^{2}\exp(\tilde{U}_{i}^{*}(H_{j},X_{i}))}.
\]
Letting $\Psi_{i}^{*}(X_{i})\equiv\Psi_{i}(H_{i}|X_{i},\mathbf{P}_{-i}^{*})$
we then generate the actions as, 
\[
H_{i}=1\{\Psi_{i1}^{*}>\omega_{i}\}
\]
where $\omega_{i}$'s are drawn iid from the uniform distribution
on $[0,1]$.

\section{Conclusion}\label{sec:conclusion_c1}

This paper presents an empirical strategy for studying wages and education
in a labour market where the decisions of workers matter in the matching
process. We demonstrate the
feasibility of our approach in the case that the worker faces a choice
between two education levels. 

One limitation of the current approach is its reliance
on cross-sectional variation alone for inference. In effect, useful
information concerning unemployment and job-to-job transitions by
workers is unused in our framework. 

This paper has also demonstrated how the decision to invest in education
- and wage inequality - is sensitive to the presence of a particular
source of matching frictions in the economy. Although firm capital
is exogenous in this paper, the role of information frictions on capital
accumulation in an extended framework could be a fruitful way to study
not only wage inequality, but also economic growth. 
%%%% DUPLICATE

\subsection*{Acknowledgements}
I thank Kyungchul Song, David Green, Vadim Marmer, and Florian Hoffmann for their advice and kind support throughout the course of this project. I have also benefited from
comments from Joris Pinske, Aureo de Paula, Hiro Kasahara, Paul Schrimpf, Michael Peters, Anna Rubinchik, Jonathan Graves, Anujit Chakraborty, and the participants of presentations at the University of British Columbia, the University of
Haifa, and the University of California (Davis).

\section{Appendix A: Equilibrium Existence and Uniqueness}\label{sec:equilibrium}

In this section, we characterize the equilibrium of the incomplete
information game of Section \ref{sec:StructuralModel}.
First, we introduce a representation of the worker's expected utility
function that proves useful for establishing the existence of the
Bayesian Nash equilibrium of the game as a fixed point of a best probability
response operator. We begin by defining relevant terms. A profile
of \textit{strategy functions} (or decision rules) is 
\[
\sigma=\{\sigma_{i}(x_{i},\varepsilon_{i}):i\in N_{h}\},
\]
where the functions $\sigma_{i}:\mathcal{X}\times\mathbf{R}^{J-1}\rightarrow\mathcal{H}$.
The conditional probability that a worker with covariates $x_{i}$
chooses action $h_{i}$ can be written

\[
P_{i}(h_{i}|x_{i},\sigma_{i})\equiv\int1\{\sigma_{i}(x_{i},\varepsilon_{i})=h_{i}\}dF(\varepsilon_{i}).
\]
Since $X_{i}$'s are private information in this model, each agent
$i$ must take expectations with respect to the distribution of $X_{-i}$.
The following result shows that under the independence assumptions
embodied by Assumption \ref{AssIndep}, the
agent's expected utility has a very convenient form - it is only affected
by the behaviour of the other agents through the choice probabilities. 
\begin{lemma}
\label{lemma_1_expected_utility_representation}In the model of Section
(\ref{sec:StructuralModel})
and Assumptions \ref{AssIndep} and \ref{AssSep}, we can represent the first term in the expected
utility of agent $i$ from equation \ref{eq:EU} as
\begin{align*}
\tilde{U}_{i}(h_{i},x_{i},\sigma) & =\sum_{h_{-i}\in\mathcal{H}_{-i}}\tilde{u}_{i}(h_{i},h_{-i},x_{i})\prod_{j\neq i}P_{j}(h_{j}|\sigma_{j}).
\end{align*}
\end{lemma}
\begin{proof}
First, we write equation \ref{eq:EU} as 
\begin{eqnarray}
\tilde{U}_{i}(h_{i},x_{i},\sigma) & = & \sum_{x_{-i}\in\mathcal{X}_{-i}}\sum_{h_{-i}\in\mathcal{H}_{-i}}\tilde{u}_{i}(h_{i},h_{-i},x_{i})P_{-i}(h_{-i}|x_{-i},\sigma)P(x_{-i}),\label{eq:lemma_1_eq_1}
\end{eqnarray}
where $x_{-i}=(x_{j})_{j\in N_{h}\backslash\{i\}}$ and we use the
shorthand $P(x_{-i})\equiv P(X_{-i}=x_{-i})$. Without loss of generality,
let $i=1$. Then we write $\tilde{U}_{1}(h_{1},x_{1},\sigma)$ as
\begin{equation}
\sum_{h_{-1}\in\mathcal{H}_{-1}}\sum_{x_{2}\in\mathcal{X}}...\sum_{x_{n}\in\mathcal{X}}\tilde{u}_{1}(h_{1},h_{-1},x_{1})P_{-1}(h_{-1}|x_{2},...,x_{n},\sigma)\prod_{j=2}^{n_{h}}P_{j}(x_{j}),\label{eq:lemma_1_eq_2}
\end{equation}

\noindent where we used the independence of $X_{i}$'s from Assumption \ref{AssIndep}. Next, since Assumption \ref{AssIndep}
says that $X_{i}$'s and $\varepsilon_{i}$'s are independent, we
know that the actions of each of the agents are independent and depend
only on their personal value of $X_{i}$ and $\varepsilon_{i}$. Therefore, 
\noindent 
\begin{eqnarray}
P(h_{-1}|x_{2},...,x_{n},\sigma) & = & \prod_{j=2}^{n_{h}}P_{j}(h_{j}|x_{2},...x_{n},\sigma_{j})=\prod_{j=2}^{n_{h}}P_{j}(h_{j}|x_{j},\sigma_{j}).\label{eq:lemma_1_eq_3}
\end{eqnarray}

\noindent Plugging (\ref{eq:lemma_1_eq_3}) back into (\ref{eq:lemma_1_eq_2})
yields that $\tilde{U}_{1}(h_{1},x_{1},\sigma)$ is equal to
\noindent 
\begin{eqnarray}
    \sum_{h_{-i}\in\mathcal{H}_{-i}}\tilde{u}_{i}(h_{i},h_{-i},x_{i})\sum_{x_{2}\in\mathcal{X}}...\sum_{x_{n}\in\mathcal{X}}\prod_{j=2}^{n_{h}}P_{j}(h_{j}|x_{j},\sigma_{j})\prod_{j=2}^{n_{h}}P_{j}(x_{j}).\label{eq:lemma_1_eq_4}
\end{eqnarray}

\noindent Grouping the sums in (\ref{eq:lemma_1_eq_4}) and restoring
the generic $i$ index gives 

\noindent 
\[
\sum_{h_{-i}\in\mathcal{H}_{-i}}\tilde{u}_{i}(h_{i},h_{-i},x_{i})\prod_{j\neq i}\sum_{x_{j}\in\mathcal{X}}P_{j}(h_{j}|X_{j}=x_{j},\sigma_{j})P_{j}(x_{j})
\]
and hence we have the desired result.
\end{proof}
We will show the existence of the equilibrium for our model. The solution
concept for the game described in Section \ref{sec:StructuralModel}
is Bayesian Nash Equilibrium (BNE), which we now define.
\begin{definition}
\textit{A Bayesian Nash Equilibrium (BNE) of the game described in
Section (\ref{sec:StructuralModel})
is a profile of decision rules $\sigma^{*}$ such that for any player
$i$ and for any $(x_{i},\varepsilon_{i})$:
\begin{eqnarray}
\sigma_{i}^{*}(x_{i},\varepsilon_{i}) & = & \text{argmax }_{h_{i}\in\mathcal{H}}\left\{ U_{i}(h_{i},x_{i},\varepsilon_{i},\sigma^{*})\right\} .\label{eq:BNE}
\end{eqnarray}
}
\end{definition}
The notation and arguments in this section follow \citet*{Aguirregabiria/Mira:19:WP},
but we include them here for completeness. Under Assumption \ref{AssSep},
we write the expected utility of $i$ as 
\begin{eqnarray*}
U_{i}(h_{i},x_{i},\sigma,\varepsilon_{i}) & = & \tilde{U}_{i}(h_{i},x_{i},\sigma)+\varepsilon_{i}'d(h_{i}).
\end{eqnarray*}
By Lemma \ref{lemma_1_expected_utility_representation} we can express
the first term on the right hand side of the the preceding equation
as
\begin{eqnarray*}
\tilde{U}_{i}(h_{i},x_{i},\sigma) & = & \sum_{h_{-i}\in\mathcal{H}_{-i}}\tilde{u}_{i}(h_{i},h_{-i},x_{i})\prod_{j\in N_{h}\backslash\{i\}}P_{j}(h_{j}|\sigma_{j}).
\end{eqnarray*}
Note that $\tilde{U}_{i}(h_{i},x_{i},\sigma)$ only depends on the
choices of other agents through the choice probabilities of the other
players that are induced by $\sigma$. We write the choice probabilities
of the people other than $i$ as
\[
\mathbf{P}_{-i}\equiv\{P_{j}(h_{j}):(j,h_{j})\in N\backslash\{i\}\times\mathcal{H}\backslash\{1\}\}.
\]
For any $\mathbf{P}_{-i}$, we can define a \textit{best response
probability function} as:
\begin{eqnarray*}
\tilde{\Psi}_{i}(h_{i}|x_{i},\mathbf{P}_{-i}) & \equiv & \int1\{\text{argmax }_{h_{i}\in\mathcal{H}}\tilde{U}_{i}(h_{i},x_{i},\sigma)+\varepsilon_{i}'d(h_{i})=h_{i}\}dF(\varepsilon_{i}).
\end{eqnarray*}
$\tilde{\Psi}_{i}$ tells us the probability that a particular action
is optimal for $i$ with covariates $x_{i}$ when others choose according
to probabilities $\mathbf{P}_{-i}$.\footnote{Note that when $\varepsilon_{i}$'s have the extreme value distribution
(as in Assumption \ref{Assumption-4}) then we have

\[
\tilde{\Psi}_{i}(h_{i}|x_{i},\mathbf{P}_{-i})=\frac{\exp(\tilde{U}_{i}(h_{i},x_{i}))}{\sum_{j=1}^{J}\exp(\tilde{U}_{i}(h_{j},x_{i}))}.
\]
} Let 
\begin{eqnarray*}
\Psi_{i}(h_{i}|\mathbf{P}_{-i}) & = & \sum_{x_{i}\in\mathcal{X}}\tilde{\Psi}_{i}(h_{i}|x_{i},\mathbf{P}_{-i})P(x_{i}).
\end{eqnarray*}
 An equivalent to Definition 1 \ref{eq:BNE} is that the equilibrium
probabilities, $\mathbf{P}^{*}\equiv\mathbf{P}(\sigma^{*})$, satisfy
the fixed point constraint, \textbf{$\mathbf{P}^{*}=\Psi(\mathbf{P}^{*})$},\textbf{
}where $\Psi$ is the best response probability mapping: 
\begin{equation}
\Psi(\mathbf{P})=\{\Psi_{i}(h_{i}|\mathbf{P}_{-i}):(i,h_{i})\in N\times\mathcal{H}\backslash\{1\}\}.\label{eq:Fixed_point}
\end{equation}
\begin{lemma}\label{existence}
Under Assumption \ref{AssIndep} and Assumption \ref{AssSep}  the game described in Section
(\ref{sec:StructuralModel}) has a Bayesian Nash Equilibrium.
\end{lemma}
\begin{proof}
Let $\mathcal{P}\equiv[0,1]^{n\times(J-1)}$. Note that $\mathcal{P}$ is a compact and convex set. Since $\Psi(\cdot)$ maps from $\mathcal{P}$ to itself and is continuously differentiable by the continuity of $\varepsilon_{i}$'s (Assumption \ref{AssIndep})), $\Psi(\cdot)$ has a fixed point by Brouwer's fixed point theorem. 
\end{proof}
We can demonstrate that the Bayesian equilibrium is unique under mild conditions on the derivatives of the best response probability mapping. Define $J_{n}(\boldsymbol{P})\equiv\frac{\partial\Psi(\boldsymbol{P})}{\partial\boldsymbol{P}'}-I_{n}$,
where $I_{n}$ is the identity matrix and let $\det(A_{n})$ denote
the determinant of an $n$-by-$n$ matrix, $A_{n}$. A result of \cite*{Kellogg:76:PAMS}, as stated in \cite*{Konovalov/Sandor:10:ET}, says that the equilibrium is unique if $\Psi$ has
no fixed points on the boundary of $\mathcal{P}$ and provided that
$\det(J_{n}(\boldsymbol{P}))$ is non-zero for each $\boldsymbol{P}\in\mathcal{P}$. 
Note that under the conditions of Lemma \ref{existence},  $J_{n}(\boldsymbol{P})$, is a matrix with $-1$'s on the diagonal and $\varphi(\boldsymbol{P})=\partial\Psi_{i}(\boldsymbol{P})/\partial p_{j}$ for all $i\neq j$ (the off-diagonals). Therefore, $J_{n}(\boldsymbol{P})$ is a circulant matrix implying the following explicit formula:
\[
\det(J_{n}(\boldsymbol{P}))=\left(\varphi(\boldsymbol{P})\cdot(n-1)-1\right)\left(-(1+\varphi(\boldsymbol{P}))\right)^{n-1}.
\]

This determinant is guaranteed to be non-zero provided that $\varphi(\boldsymbol{P})\neq1/(n-1)$ for every $\boldsymbol{P}\in\mathcal{P}$ and $n\geq2$. In our setup, the requirement that $\Psi$ have no fixed points on the boudnary of $\mathcal{P}$ holds under weak conditions on the distribution of $\varepsilon_{i}'$s. 

\section{Appendix B: First-Stage Estimation of $\theta_0$}\label{sec:AppendixB}
\subsection{Characterization of Matching Probabilities}
The remaining results of this section allow us to represent the matching
probabilities from equation \ref{eq:section_3_pi_m_j_i}, hence workers'
expectations, in a convenient way. These representations can then
be used to estimate $\theta(\beta)$ using maximum likelihood. 
\begin{proposition}
\label{Proposition_1}Suppose that Assumptions \ref{AssIndep}, \ref{AssSep},
\ref{AssInfo}, \ref{Assumption-4} hold, and  that $J=2$. Then the probability that any worker $i$  matches to
a firm from capital class $m$ conditional on choosing education level $h_{j}$ is 
\begin{eqnarray*}
\pi_{mj} & =\sum_{n_{j}=0}^{n_{h}-1} & P(\mathcal{M}(i)=m|H_{i}=h_{j},N_{j}=n_{j})B(n_{j};n_{h}-1,p_{j}),
\end{eqnarray*}
 for each $m=1,...,M$ and each $h_{j}\in\mathcal{H}$, where $N_{j}$ is the number of workers other than $i$ who picked
education level $h_{j}$, $B(n_{j};n_{h}-1,p_{j})$ is the binomial p.m.f.
and $p_{j}=P(H_{i}=h_{j}).$
\end{proposition}
\begin{proof}
Part (a) of Assumption \ref{AssInfo} that says firms do not consider the workers'
covariates when ranking them in the matching process. This means that
for each $m=1,...,M$ we have that 
\begin{align*}
P(\mathcal{M}(i)=m|H_{i}=h_{j},H_{-i}=h_{-i},X_{i}=x_{i}) & =P(\mathcal{M}(i)=m|H_{i}=h_{j},H_{-i}=h_{-i}).
\end{align*}
Combining this with equation \ref{eq:section_3_pi_m_j_i}, we can
write
\begin{eqnarray*}
\pi_{mj}^{(i)} & = & \sum_{h_{-i}\in\mathcal{H}_{-i}}P(\mathcal{M}(i)=m|H_{i}=h_{j},H_{-i}=h_{-i})P(H_{-i}=h_{-i}|X_{i}=x_{i}).
\end{eqnarray*}
Next, it is straightfoward to see that\footnote{This can be shown using the same arguments as those in Lemma \ref{lemma_1_expected_utility_representation}.
The private information and independence of $X_{i}$'s (Assumption
\ref{AssIndep}) implies that the left hand size of \ref{eq:prodmarginal} equals
\begin{eqnarray}
\sum_{x_{-i}\in\mathcal{X}_{-i}}P(h_{-i}|x_{-i},x_{i})P(x_{-i}|x_{i}) & = & \sum_{x_{-i}\in\mathcal{X}_{-i}}P(h_{-i}|x_{-i})P(x_{-i}).\label{eq:Proposition_1_eq_1}
\end{eqnarray}

Suppose without loss of generality that $i=1$. It is convenient to
rewrite the above as follows (using independence):
\[
\sum_{x_{2}\in\mathcal{X}}...\sum_{x_{n}\in\mathcal{X}}P(h_{-1}|x_{2},...,x_{n})\prod_{j=2}^{n_{h}}P_{j}(x_{j}).
\]
Next, since $X_{i}$'s and $\varepsilon_{i}$'s are independent across
$i$ and each $i$'s strategy function is only a function of $X_{i}$
and $\varepsilon_{i}$ we have
\begin{eqnarray*}
P(h_{-1}|x_{-1})=\prod_{j\neq1}^{n_{h}}P_{j}(h_{j}|x_{-1}) & = & \prod_{j\neq1}^{n_{h}}P_{j}(h_{j}|x_{j}).
\end{eqnarray*}
Combining these two results we write \ref{eq:Proposition_1_eq_1}
as
\begin{align*}
\sum_{x_{2}\in\mathcal{X}}P_{j}(h_{2}|x_{2})P_{j}(x_{2})...\sum_{x_{n}\in\mathcal{X}}P_{n}(h_{n}|x_{n})P_{n}(x_{n}) & =\prod_{j\neq1}^{n_{h}}P_{j}(h_{j}).
\end{align*}
}
\begin{equation}
P(H_{-i}=h_{-i}|X_{i}=x_{i})=\prod_{j\neq i}^{n_{h}}P_{j}(h_{j}).\label{eq:prodmarginal}
\end{equation}
 Since $\varepsilon_{i}$ are identically distributed by Assumption \ref{AssIndep}, for
each $j$ and $m$, have $\pi_{mj}^{(i)}=\pi_{mj}$. 

When there are only two education levels, any $h_{-i}\in\mathcal{H}_{-i}$
can be represented as a total number of workers other than $i$ who
picked education level $h_{j}$, $n_{j}$. From worker $i$'s point of
view, $n_{j}$ is a particular realization of the random variable
$N_{j}$ that takes values in the set $\{0,...,n_{h}-1\}$. Since
there are $n_{h}-1$ agents other than $i$ in the economy, the sum
over $h_{-i}\in\mathcal{H}_{-i}$ amounts to a sum over the support
of $N_{j}$. Now consider any $n_{j}$ in the support of $N_{j}$.
The assumption that $\varepsilon_{i}$'s are iid implies that the
probability that exactly $n_{j}$ out of $n_{h}-1$ workers pick $h_{j}$
can be represented as
\[
\frac{(n_{h}-1)!}{n_{j}!(n_{h}-1-n_{j})!}p_{j}^{n_{j}}(1-p_{j})^{n_{h}-1-n_{j}},
\]
which is the binomial probability mass function, $B(n_{j};n_{h}-1,p_{j})$. 
\end{proof}

When $J=2$ we can partition the types of firms, $m=1,...,M$ into
two sets: those who prefer $h_{j}\in\mathcal{H}$ and those who prefer $h_{j'}$ with
$j'\neq j$. It is convenient to introduce the following notation:
\begin{align}
M_{j}^{+}(\theta) & =\{m\in\{1,...,M\}:\rho(k_{m},h_{j};\theta)\geq\rho(k_{m},h_{j'};\theta),j\neq j'\},\text{ and}\\
M_{j}^{-}(\theta) & =\{1,...,M\}\backslash M_{j}^{+},\label{eq:M_sets}
\end{align}

\noindent recalling that firm preferences are given in \ref{eq:firm_preferences}.
The firm classes that prefer $h_{j}$ are pinned down by the functional
form for firm preferences, $\rho$,  the preference parameter,
$\theta$, and the distribution of $X_{i}$. Furthermore let us denote
\begin{eqnarray}
\pi_{mj} & \equiv & \sum_{n^{(j)}=0}^{n_{f}}\sum_{n_{j}=0}^{n_{h}-1}P_{h_{j},n_{j},n^{(j)}}(m)P(n_{j})P(n^{(j)};\theta),\label{eq:pimjtheta}
\end{eqnarray}
where 
\begin{align}
P_{h_{j},n_{j},n^{(j)}}(m) & =P(\mathcal{M}(i)=m|h_{i}=h_{j},N_{j}=n_{j},N^{(j)}=n^{(j)}).\label{eq:Section_8_pt_2_pjnjnj}
\end{align}

Note that this object depends on both $\beta$ and $\theta$ through the matching function.
For each firm type $m=1,...M$ let $F_{m}\equiv N(\beta k_{m},\sigma^{2})$,
and define the following for each education choice $h_{j}$: 
\[
G_{j+}\equiv\sum_{m\in M_{j}^{+}}q_{m}F_{m}\text{ and }G_{j-}\equiv\sum_{m\in M_{j}^{-}}q_{m}F_{m}.
\]
Furthermore, define the posterior firm types as follows:
\[
q_{m}^{+}\equiv q_{m}/\sum_{m\in M_{j}^{+}}q_{m}\text{ and }q_{m}^{-}\equiv q_{m}/\sum_{m\in M_{j}^{-}}q_{m}.
\]
We also define $v_{(b_{1},b_{2};F)}$ as the $b_{1}$-order statistic
of $b_{2}$ random variables independently distributed according to
cdf $F$. Propositions \ref{Proposition_2_heterogeneous_firms} and
\ref{Proposition_3_homogeneous_firms} are characterizations of $P_{h_{j},n_{j},n^{(j)}}(m)$'s
of the model in the case that $J=2$ and $n_{h}=n_{f}=n$. 

When considering these results, it is important to recall one core
feature of the matching model as we outline it in Section \ref{sec:section_two_structural_model}:
that there is no unemployment. Therefore, when reading the arguments,
the reader should take for granted the fact that the probability that
each worker matches to some firm occurs with probability one. 
\begin{proposition}
\noindent (Heterogeneous firm preferences).\label{Proposition_2_heterogeneous_firms}
Denote $\bar{n}_{j}\equiv n_{j}+1$ and suppose that $n_{h}=n_{f}=n$.
Then under the assumptions of Proposition \ref{Proposition_1} we have the following
for any $n_{j}$ such that $1\leq n_{j}\leq n$ and $n^{(j)}$ such
that $0<n^{(j)}<n$:

i) For each $m\in M_{j}^{+}$,
\[
P_{h_{j},n_{j},n^{(j)}}(m)=\begin{cases}
q_{m}^{+}n^{(j)}/\bar{n}_{j} & \text{ if }\bar{n}_{j}\geq n^{(j)}\\
\frac{P(v_{m}>\hat{v})q_{m}^{+}}{\sum_{m\in M_{j}^{+}}P(v_{m}>\hat{v})q_{m}^{+}} & \text{ if }\bar{n}_{j}<n^{(j)}
\end{cases},
\]
where $\hat{v}\equiv v_{(a,b;F)}$ with $a=n^{(j)}-\bar{n}_{j}$,
$b=n^{(j)}$, and $F=G_{j+}$. 

ii) For each $m\in M_{j}^{-}$, 
\[
P_{h_{j},n_{j},n^{(j)}}(m)=\begin{cases}
\frac{P(v_{m}<\hat{v})q_{m}^{-}\left((\bar{n}_{j}-n^{(j)})/\bar{n}_{j}\right)}{\sum_{m\in M_{j}^{-}}P(v_{m}<\hat{v})q_{m}^{-}} & \text{ if }\bar{n}_{j}>n^{(j)}\\
0 & \text{ if }\bar{n}_{j}\leq n^{(j)}
\end{cases},
\]
where $\hat{v}\equiv v_{(a,b;F)}$ with $a=\bar{n}_{j}-n^{(j)}+1$,
$b=n-n^{(j)}$, and $F=G_{j-}$. 

\end{proposition}
\begin{proof}
We begin by introducing some notation. We denote the event that a
worker $i$ who chose education level $h_{j}$ matches to any firm of type
$m\in M_{j}^{+}$ or $m\in M_{j}^{-}$ as $M_{ij}^{+}$ and $M_{ij}^{-}$
respectively.\footnote{That is, $M_{ij}^{+}\equiv\{\mathcal{M}_{i}\in M_{j}^{+}\}$ and similarly
for $M_{ij}^{-}\equiv\{\mathcal{M}_{i}\in M_{j}^{+}\}$.} 

First, we consider the probability that a worker who chose $h_{j}$ matches
to any firm in the class $m\in M_{j}^{+}.$ Consider the case that
$\bar{n}_{j}\geq n^{(j)}$. In this case, there are at least as many
workers who chose $h_{j}$ as firms who prefer $h_{j}$. Given that Condition
IR implies that no worker or firm will never unilaterally dissolve
a match to become unmatched, the case of $\bar{n}_{j}\geq n^{(j)}$
implies that every firm in class $m$ who wants a worker with $h_{j}$
will hire one in the matching process. For each class of firm $m\in M_{j}^{+}$,
the probability that a worker who chose $h_{j}$ matches to a firm in
the set of firms that prefers $h_{j}$ and to the particular class $m\in M_{j}^{+}$
is given as follows when $\bar{n}_{j}\geq n^{(j)}$:
\begin{eqnarray*}
P_{j}(\mathcal{M}_{i}=m,M_{ij}^{+}) & = & P_{j}(\mathcal{M}_{i}=m|M_{ij}^{+})P_{j}(M_{ij}^{+}).\\
 & = & q_{m}^{+}n^{(j)}/\bar{n}_{j},
\end{eqnarray*}
where the $j$-subscript on the probabilities denote a probability
conditional on the event $H_{i}=h_{j}$. $P_{j}(M_{ij}^{+})$ is equal to $n^{(j)}/\bar{n}_{j}$ because workers
with the same $h_{j}$ are indistinguishable to the firms that prefer
them, so firms choose among these workers at random. The probability
of matching to a firm of type $m\in M_{j}^{+}$ given that the worker
has already matched to some firm in $M_{j}^{+}$ is equal to the relative
proportion of type $m$ firms in this category, $q_{m}^{+}$. 

Next, we consider the case that $\bar{n}_{j}<n^{(j)}.$ Since there
are strictly more firms that prefer $h_{j}$ than workers who chose $h_{j}$,
the probability that a worker who chose $h_{j}$ matches to a firm that
prefers workers with $h_{j}$ occurs with probability one; that is $P_{j}(M_{ij}^{+})=1$.\footnote{This follows from Condition IR and the following two facts: i) $h_{j}$
workers are scarce relative to the firms that prefer them ii) firms
that prefer $h_{j'}$ will never choose a $h_{j}$ worker in the matching
process since the condition $n_{h}=n_{f}=n$ and $J=2$ implies that
$h_{j'}$ workers are always available (i.e., when $n_{h}=n_{f}=n$, $n^{(j)}>\bar{n}_{j}$
implies that $n_{j'}>n^{(j')}$, since $n_{j'}=n-\bar{n}_{j}$ and
$n^{(j')}=n-n^{(j)}$).}

Although $P_{j}(M_{ij}^{+})=1$, only the firms with the $\bar{n}_{j}$
largest $v$-indices will be able to match with a worker who chose
$h_{j}$. Thus, a firm in $M_{j}^{+}$ matches to a worker with $h_{j}$ if
and only if its $v$ statistic exceeded the $\kappa=n^{(j)}-\bar{n}_{j}$
order statistic among all $n^{(j)}$ firms in $M_{j}^{+}$. Thus,
by Assumptions \ref{AssInfo} and \ref{Assumption-4}, the probability that a worker who chose $h_{j}$
matches with a firm from a particular class $m\in M_{j}^{+}$ conditional on matching
to some firm in $M_{j}^{+}$ is 
\[
P(v(K)=v(k_{m})|v(K)>\hat{v},m\in M_{j}^{+}),
\]
which by Bayes' rule equals
\begin{align}
 & \frac{P(v(K)>\hat{v}|v(K)=v(k_{m}),m\in M_{j}^{+})P(v(K)=v(k_{m})|m\in M_{j}^{+})}{\sum_{m\in M_{j}^{+}}P(v(K)>\hat{v}|v(K)=v(k_{m}),m\in M_{j}^{+})P(v(K)=v(k_{m})|m\in M_{j}^{+})},\label{eq:threshhold_crossing}
\end{align}
where $\hat{v}\equiv v_{(\kappa,n^{(j)};G_{j+})}$. Equation \ref{eq:threshhold_crossing}
represents the relative proportion of type $m$ firms represented
among threshhold crossers among all firms that prefer $h_{j}$. We next
consider the probability of matching to each firm with $m\in M_{j}^{-}$.
We consider first the case that $\bar{n}_{j}>n^{(j)}$. The relevant
probability is
\begin{eqnarray*}
P_{j}(\mathcal{M}_{i}=m,M_{ij}^{-}) & = & P_{j}(\mathcal{M}_{i}=m|M_{ij}^{-})P_{j}(M_{ij}^{-})\\
 & = & P_{j}(\mathcal{M}_{i}=m|M_{ij}^{-})(1-n^{(j)}/\bar{n}_{j}).
\end{eqnarray*}
As stated above, the of case $\bar{n}_{j}>n^{(j)}$ combined with
our assumption that $n_{h}=n_{f}=n$ implies that $n^{(j')}>n_{j'}$,
since $n_{j'}=n-\bar{n}_{j}$ and $n^{(j')}=n-n^{(j)}$. Therefore
by similar logic to before, firms who prefer $h_{j'}$ match to workers
with $h_{j}$ if their $v$-index is lower than the $n^{(j')}-n_{j'}+1=\bar{n}_{j}-n^{(j)}+1$
order statistic among those firms in $M_{j}^{-}$. Letting $\kappa\equiv\bar{n}_{j}-n^{(j)}+1$, the probability of a worker who chose $h_{j}$ matching to a type $m\in M_{j}^{-}$
firm conditional on matching to some firm in $M_{ij}^{-}$ is given
as the proportion of type $m$ firms whose $v$ index falls below
this threshhold: 

\[
P_{j}(\mathcal{M}_{i}=m|M_{ij}^{-})=\frac{P(v_{m}<\hat{v})q_{m}^{-}}{\sum_{m\in M_{j}^{-}}P(v_{m}<\hat{v})q{}_{m}^{-}},
\]
where $\hat{v}\equiv v_{(\kappa,n^{(j')};G_{j-})}$. Lastly, in the
case that $\bar{n}_{j}\leq n^{(j)}$, $P(M_{ij}^{-})=0$. This completes
the proof. 
\end{proof}
Next we define $G\equiv\sum_{m=1}^{M}F_{m}q_{m}$. Proposition \ref{Proposition_3_homogeneous_firms}
characterizes the matching probabilities in the case that all firms
types prefer one level of education; that is, in the case that firm
preferences are homogeneous over worker education types. The arguments
are abridged, since they are very similar to those used in the proof
of Proposition \ref{Proposition_2_heterogeneous_firms}. 
\begin{proposition}
\noindent (Homogeneous firm preferences).\label{Proposition_3_homogeneous_firms}
Suppose that $n_{h}=n_{f}=n$. Then under the assumptions of Proposition \ref{Proposition_1} we have the following for the cases that $n^{(j)}=n$ and $n^{(j)}=0.$ 
\begin{enumerate}
\item if $n^{(j)}=n$, then $M_{j}^{-}=\emptyset$ and for each $m\in M_{j}^{+}=M$
we have 

\[
P_{h_{j},n_{j},n^{(j)}}(m)=\begin{cases}
q_{m} & \text{ if }\bar{n}_{j}=n\\
\frac{P(v_{m}>\hat{v})q_{m}}{\sum_{m\in M}P(v_{m}>\hat{v})q_{m}} & \text{ if }\bar{n}_{j}<n
\end{cases},
\]
where $\hat{v}\equiv v(a,b;G)$, with $a=n-\bar{n}_{j}$ and
$b=n$. 
\item If $n^{(j)}=0$, then $M_{j}^{+}=\emptyset$ and for each $m\in M_{j}^{-}=M$
we have 

\begin{align*}
P_{h_{j},n_{j},n^{(j)}}(m) & =\begin{cases}
q_{m} & \text{ if }\bar{n}_{j}=n\\
\frac{P(v_{m}<\hat{v})q_{m}}{\sum_{m\in M}P(v_{m}<\hat{v})q_{m}} & \text{if }\bar{n}_{j}<n
\end{cases},
\end{align*}
where $\hat{v}\equiv v(a,b;G)$, with $a=\bar{n}_{j}+1$ and
$b=n$. 
\end{enumerate}
\end{proposition}
\begin{proof}
When $n^{(j)}=n$ and $\bar{n}_{j}=n$ the probability of matching
to firm $m$ is simply equal to the marginal probability of that firm
type in the economy, $q_{m}$. When $n^{(j)}=n$ and $\bar{n}_{j}<n$,
using logic identical to that employed in the proof of Proposition
\ref{Proposition_2_heterogeneous_firms}, we conclude that the probability of matching to a firm from class
$m$ is equal to the proportion of type $m$ firms above the $n-\bar{n}_{j}$
order statistic of the $v$'s.

When $n^{(j)}=0$, we must have $\bar{n}_{j}>n^{(j)}=0$ (since at
least one person is assumed to choose $h_{j}$). Since the top $n_{j'}=n-\bar{n}_{j}$
ranked firms in terms of $v$ receive a worker with their preferred
education, $h_{j'}$, the probability of matching to a firm in class $m$
is equal to the proportion of type $m$ below the $\bar{n}_{j}+1$
order statistic of the $v$'s. 
\end{proof}
The following result takes for granted a well-known fact that uniform
order statistics follow the Beta distribution.\footnote{For example, see Chapter 2 \citet*{Ahsanullah/Nevzorov/Shakil:13:Springer}. }
\begin{lemma}
\label{Lemma_4_order_statistic_lemma}Let: i) $\{X_{i}\}_{i=1}^{n}$
be iid random variables from continuous distribution function $G$;
ii) $Z$ be normally distributed with mean $\mu$ and variance $\sigma^{2}$;
iii) $X_{(i)}$ be the $i$-th order statistic of $\{X_{i}\}_{i=1}^{n}$;
iv) $U_{(i)}$ be the $i$-th order statistic of iid uniform random
variables $\{U_{i}\}_{i=1}^{n}$. Then,
\[
P(Z\geq X_{(i)})=1-\mathbf{E}\Phi((G^{-1}(U_{(i)})-\mu)/\sigma),
\]
where $\Phi(\cdot)$ is the standard normal cdf, and $\mathbf{E}(\cdot)$
is taken over the distribution of $U_{(i)}$, which follows the Beta
distribution with parameters $i$ and $n+1-i$. 
\end{lemma}
\begin{proof}
Note that since $X_{i}$'s are continuously distributed according
to $G$ it follows from the probability integral transformation result
that for each $i$
\[
X_{i}=_{d}G^{-1}(U_{i}).
\]
Also, since $G$ is monotone we have that for each $i$
\[
X_{(i)}=_{d}G^{-1}(U_{(i)}).
\]
The previous line implies that 
\begin{eqnarray*}
P(Z\geq X_{(i)}) & = & P(Z\geq G^{-1}(U_{(i)}))\\
 & = & 1-P(Z\leq G^{-1}(U_{(i)}))\\
 & = & 1-\mathbf{E}\Phi((G^{-1}(U_{(i)})-\mu)/\sigma),
\end{eqnarray*}
where $\mathbf{E}(\cdot)$ is taken over the distribution of $U_{(i)}$.
The last equality used the fact that $Z$ is normal with mean $\mu$
and variance $\sigma^{2}$. 
\end{proof}
The following results are a direct application of the previous results.
They are useful for constructing the $\pi_{mj}$'s that are used in
the structural estimation of this paper. Recall the definitions of
$G$, $G_{j+}$, $G_{j-}$, and $v(b_{1},b_{2},F)$ from before. We
introduce the following notation: 
\[
a(\kappa,n,m;G)\equiv\mathbf{E}\Phi\left((G^{-1}(U_{(\kappa;n)})-\beta k_{m})/\sigma_{m}\right),
\]
where $U_{(\kappa;n)}$ is the $\kappa$-order statistic of $n$ uniform
random variables and $\mathbf{E}(\cdot)$ is taken over the distribution
of $U_{(\kappa;n)}$. 
\begin{corollary}\label{corollary:1}
Suppose the conditions of Proposition \ref{Proposition_2_heterogeneous_firms} hold and let $v_{m}$ be distributed
according to $F_{m}$. Then, in the heterogeneous preferences case
with $\bar{n}_{j}<n^{(j)}$, 
\begin{enumerate}
\item \begin{flushleft}
For each $m\in M_{j}^{+}$, $P(v_{m}>v_{(\kappa,n^{(j)};G_{j+})})=1-a(\kappa,n^{(j)},m;G_{j+}),$
where $\kappa=n^{(j)}-\bar{n}_{j}$.
\par\end{flushleft}
\item \begin{flushleft}
For each $m\in M_{j}^{-}$, $P(v_{m}<v_{(\kappa;n^{(j')};G_{j-})})=a(\kappa,n^{(j')},m;G_{j-}),$
where $\kappa=\bar{n}_{j}-n^{(j)}+1$. 
\par\end{flushleft}

\end{enumerate}
\end{corollary}
\begin{corollary}\label{corollary:2}
Suppose the conditions of Proposition \ref{Proposition_3_homogeneous_firms} hold and let $v_{m}$ be distributed
according to $F_{m}$. Then, in the homogeneous preferences case with
$\bar{n}_{j}<n$, 
\begin{enumerate}
\item \begin{flushleft}
If $n^{(j)}=0$, $P(v_{m}<v_{(\kappa;n,G)})=a(\kappa,n,m;G)$ for
each $m\in M$, where $\kappa=\bar{n}_{j}+1$.
\par\end{flushleft}
\item \begin{flushleft}
If $n^{(j)}=n$, $P(v_{m}>v_{(\kappa;n,G)})=1-a(\kappa,n,m;G)$ for
each $m\in M$, where $\kappa=n-\bar{n}_{j}$.
\par\end{flushleft}
\end{enumerate}
\end{corollary}
\begin{proof}
The proofs of Corollaries 1 and 2 follows directly from Lemma \ref{Lemma_4_order_statistic_lemma}.
\end{proof}

\subsection{Identification of $\theta_0$}\label{sec:identification}

In extremum estimation problems, it is common to argue that the parameter is identified via a set of sufficient
conditions requiring continuity (or semi-continuity) of the underlying population criterion function. See for example, \citet*{Newey/McFadden:94:HE}.  In our context, assuming continuity of this criterion function is
inappropriate, since the probabilities that workers match to certain
firm types conditional on their choice of education are discontinuous
in $\theta$. The problem arises because these probabilities depend
on the sets of firm types that prefer high and low education,
which are discontinuous functions of $\theta$.\footnote{Recall the definitions in equation \ref{eq:M_sets}.} Under the model of Section \ref{sec:StructuralModel} we provided assumptions under
which it was natural to estimate $\theta$ using maximum likelihood.
In this section, we show that the discontinuity of the population
likelihood function is not an obstacle to identification. As we will see, the required conditions are only slightly stronger than those commonly used to identify the standard logit model. 

Before presenting the identification result and the proof, we define
terms. The population objective function is
\begin{align*}
Q(\theta) & =\mathbf{E}\log\ell_{i}(\theta)=\mathbf{E}\left[h_{i}\log p_{i}(\theta)+(1-h_{i})\log\left(1-p_{i}(\theta)\right)\right],
\end{align*}
where $p_{i}(\theta)=\Lambda(\Delta_{i}(\theta)),$ $\Lambda(z)=\exp(z)/(1+\exp(z))$,
and $\Delta_{i}(\theta)=\omega_{1i}(\theta,\beta_{0})-\omega_{0i}(\theta,\beta_{0})$
are the population version of the objects in equation \ref{eq:expected_age}.\footnote{Note that our notation in this section will occasionally omit $\beta_{0}$
for the purposes of clarity. } Let us also define $\psi(X_{i};\theta)\equiv g(h_{1},X_{i};\theta)-g(h_{0},X_{i};\theta)$
and write the probability of high education as
\begin{align*}
p_{i}(\theta) & =\Lambda\left(\tau\cdot(f_{1}(\theta,\beta_{0})-f_{0}(\theta,\beta_{0}))+(1-\tau)\cdot\psi(X_{i};\theta)\right)\\
 & =\Lambda\left(\tau\cdot\Phi(\theta)'\Pi(\theta)+(1-\tau)\cdot\psi(X_{i};\theta)\right),
\end{align*}
where 
\begin{align*}
\Pi(\theta) & \equiv(\pi_{1}(\theta,\beta_{0}),\pi_{0}(\theta,\beta_{0}))'\in\mathbf{R}^{2M\times1},\\
\Phi(\theta) & \equiv\begin{array}{cc}
(\phi_{1}(\theta), & -\phi_{0}(\theta)\end{array})'\in\mathbf{R}^{2M\times1},
\end{align*}
and $\phi_{j}(\theta)'\equiv(f(h_{j},k_{1};\theta),...,f(h_{j},k_{M};\theta))'\in\mathbf{R}^{M\times1},\text{ for }j=0,1.$
We are now ready to present the main result of this section.
\begin{theorem}\label{thm:theta_id}
 $\theta_{0}$ is identified up to $\beta_{0}$ and uniquely maximizes $Q(\theta)$ over
$\theta\in\Theta$ if
\begin{enumerate}
\item For each $\theta\in\Theta$ such that $\theta\neq\theta_{0}$, $\bar{\psi}_{i}(\theta,\theta_{0})\equiv\psi(X_{i},\theta)-\psi(X_{i},\theta_{0})$
is continuously distributed, and
\item for each $\theta\in\Theta$, $|\Phi(\theta)\Pi(\theta)|<\infty$ and
$\mathbf{E}|\psi(X_{i};\theta)|<\infty$.
\end{enumerate}
\end{theorem}

\begin{proof} For this result, we will use the well-known fact that in maximum likelihood problems,
identification implies unique maximization.\footnote{For example, See Lemma 2.2 of \citet*{Newey/McFadden:94:HE}}
$\theta_{0}$ is identified and $Q(\theta)$ has a unique maximum
at $\theta_{0}$ if for all $\theta\in\Theta$ with $\theta\neq\theta_{0}$
\begin{align}
P(\ell_{i}(\theta)\neq\ell_{i}(\theta_{0}))>0, & \text{ and }\label{eq:idc1}\\
\mathbf{E}|\log\ell_{i}(\theta)|<\infty.\label{eq:idc2}
\end{align}
We begin by showing that equation \ref{eq:idc1} holds when $\bar{\psi}_{i}(\theta,\theta_{0})$
is continuously distributed. Since $\Lambda(z)$ and $1-\Lambda(z)$ are strictly monotonic in
$z$, we will show equation \ref{eq:idc1} by proving that for all $\theta\neq\theta_{0}$ 
\begin{equation}
P\left(\Delta_{i}(\theta)\neq\Delta_{i}(\theta_{0})\right)>0\label{eq:idc1_1}.
\end{equation}

To establish this, we define some notation. Let $M(\theta)$ refer to the set of
firm types that prefers high education.\footnote{Our apparent focus on firm types that prefer high education is without
loss of generality since any $M(\theta)$ also uniquely defines an
associated set of firm types that prefers low education.} Define the relation $\sim$ on $\Theta$ to be $\theta\sim\theta'$
if and only if $M(\theta)=M(\theta')$. Since $\sim$ is an equivalence
relation on $\Theta$, it can be partitioned into
a union of $S$ equivalence classes, $\Theta=\cup_{s=1}^{S}\Theta_{s}$,
where $S$ is equal to the number of firm types plus one. Note that
the only way in which $\theta$ affects $\Pi(\theta)$ is through
$M(\theta)$. Hence $c_{s}=\Pi(\theta)$ is constant over $\theta\in\Theta_{s}$
on each partition $s=1,...,S$. Therefore, it follows that \ref{eq:idc1_1}
holds since for all $\theta\in\Theta\backslash\{\theta_{0}\}$ and
each $s$, we have that
\begin{align}
P\left(\bar{\psi}_{i}(\theta,\theta_{0})\neq\frac{\tau}{1-\tau}\left(\Phi(\theta)'c_{s}-\Phi(\theta_{0})'c_{0}\right)\right)=1.\label{eq:for_idc1}
\end{align}
Equation \ref{eq:for_idc1} is true because $\bar{\psi}_{i}$ is continuously distributed and
the right-hand-side of the inequality is non-stochastic. Lastly, we show that equation \ref{eq:idc2} holds. Note that 
\begin{align*}
\mathbf{}|\log\ell_{i}(\theta)| & \leq2\times|\log\Lambda(\Delta_{i}(\theta))|\\
 & \leq2\times\left(|\log\Lambda(0)|+|\Delta_{i}(\theta)|\right).
\end{align*}
Taking expectations of both sides and applying the second condition
gives the result.
\end{proof}
The requirement that $\mathbf{E}|\psi(X_{i};\theta)|<\infty$ follows
under familiar circumstances. For example, when we take an outside option of $g(h,x;\theta)=h\cdot x'\theta$, the condition follows from the existence of finite second moments. Our condition
is more general so as to allow additional flexibility in the choice
of $g$. The existence of $|\Phi(\theta)\Pi(\theta)|$ is also satisfied
under very mild assumptions. Since $\pi_{1}(\theta,\beta_{0})$ and
$\pi_{0}(\theta,\beta_{0})$ are always bounded between zero and one,
a sufficient condition for the existence of $|\Phi(\theta)\Pi(\theta)|$
is that $f$ be such that $f^{2}(h,k;\theta)<\infty$ for each $\theta\in\Theta$,
$h$, and $k$.

The following result establishes a set of sufficient assumptions for
the first condition of Theorem \ref{thm:theta_id}.
\begin{lemma}
Suppose that (a) $X_{i}$ is continuously distributed and (b) the
outside option function $g$ is such that (i) $g(h,x;\theta)=g(h,x'\theta)$
and (ii) $g(h_{1},z)-g(h_{0},z)$ is strictly monotonic in $z$. Then
$\bar{\psi}_{i}(\theta,\theta_{0})$ is continuously distributed for
all $\theta\neq\theta_{0}$. 
\end{lemma}

\begin{proof}
Since $X_{i}$ is continuously distributed and $g(h_{1},z)-g(h_{0},z)$
is strictly monotonic in $z$, then $\psi_{i}(X_{i};\theta)=g(h_{1},X_{i}'\theta)-g(h_{0},X_{i}'\theta)$
is itself continuously distributed. Thus, for all $\theta\neq\theta_{0}$
it follows that $\bar{\psi}_{i}(\theta,\theta_{0})=\psi_{i}(X_{i};\theta)-\psi_{i}(X_{i};\theta_{0})$
is also continuously distributed. 
\end{proof}
For example, it is clear that linear outside option functions satisfy the conditions of the above lemma. An example of a non-linear outside option function satisfying the conditions is $g(h,x'\theta)=\exp(h\cdot x'\theta)$.\footnote{One way in which Theorem \ref{thm:theta_id} appears strong is that it implicitly assumes that all elements of $\theta$ enter the outside option. However, the simulation study presented in Section \ref{sec:MonteCarlo} offers evidence that $\hat{\theta}_{n}(\beta_{0})$  remains consistent when this requirement is relaxed somewhat. In particular, the simulation demonstrates that the inference performs well when a single element of $\theta_{0}$ enters the production function but not the outside option.}

%which we use in the simulations study of Section $\ref{sec:MonteCarlo}$.

Consistency of $\hat{\theta}_{n}(\beta_{0})$ holds when the data
are iid and $\Theta$ is compact under Theorem 2.5. of \citet*{Newey/McFadden:94:HE}, under the additional requirement that the
likelihood function be continuous at each $\theta\in\Theta$ with
probability one. As pointed out by the authors, the latter condition is mild
in the sense that it does not require the likelihood be continuous
at every $\theta$ for a given realization of the random variables.
In our context, this requirement is satisfied if the probability that
a firm is indifferent between a worker with high and low education
is zero for every value of the preference parameter, $\theta$. Asymptotic normality of $\hat{\theta}_{n}(\beta)$ then follows provided the conditions of Theorem 3.3 of  \citet*{Newey/McFadden:94:HE} hold. One of the requirements of this asymptotic normality theorem is that the population density be twice differentiability in $\theta$ within a \textit{neighborhood} of $\theta_0$.  Although the density in our context is not guaranteed to be differentiable for every $\theta\in\Theta$, it may nonetheless satisfy this much weaker requirement of twice differentiability in a neighborhood of the true parameter.

\bibliographystyle{agsm}
\bibliography{references}

\end{document}